\documentclass{article}
\usepackage[utf8]{inputenc}
\usepackage{tikz,tikz-cd}
\usepackage[normalem]{ulem}
\usepackage{amssymb}
\usepackage{relsize}
\usepackage{amsmath}
\usepackage{mathrsfs}
\usepackage{arydshln}
\usepackage{mathtools}
\usepackage{easybmat}
\usepackage{amsthm}
\numberwithin{equation}{section}
\usepackage{bbm}
\usepackage{arydshln}
\usepackage{soul}
\usepackage{framed}

\usepackage{color}
\usepackage[all]{xy}
\usepackage{stmaryrd}

\usepackage{scalerel} 	

\usepackage{tikz}
\usetikzlibrary{svg.path}

\definecolor{orcidlogocol}{HTML}{A6CE39}
\tikzset{
  orcidlogo/.pic={
    \fill[orcidlogocol] svg{M256,128c0,70.7-57.3,128-128,128C57.3,256,0,198.7,0,128C0,57.3,57.3,0,128,0C198.7,0,256,57.3,256,128z};
    \fill[white] svg{M86.3,186.2H70.9V79.1h15.4v48.4V186.2z}
                 svg{M108.9,79.1h41.6c39.6,0,57,28.3,57,53.6c0,27.5-21.5,53.6-56.8,53.6h-41.8V79.1z M124.3,172.4h24.5c34.9,0,42.9-26.5,42.9-39.7c0-21.5-13.7-39.7-43.7-39.7h-23.7V172.4z}
                 svg{M88.7,56.8c0,5.5-4.5,10.1-10.1,10.1c-5.6,0-10.1-4.6-10.1-10.1c0-5.6,4.5-10.1,10.1-10.1C84.2,46.7,88.7,51.3,88.7,56.8z};
  }
}

\newcommand\orcidicon[1]{\href{https://orcid.org/#1}{\mbox{\scalerel*{
\begin{tikzpicture}[yscale=-1,transform shape]
\pic{orcidlogo};
\end{tikzpicture}
}{|}}}}

\usepackage[draft=false,backref,breaklinks]{hyperref} 

\newtheorem{Def}{Definition}[section]
\newtheorem{Thm}{Theorem}[Def]
\newtheorem{Cor}{Corollary}[Def]
\newtheorem{Lem}{Lemma}[Def]
\newtheorem{Prop}{Proposition}[Def]
\newtheorem{Rem}{Remark}[Def] 
\newtheorem{Ex}{Example}[Def]
\newcommand\redsout{\bgroup\markoverwith{\textcolor{red}{\rule[0.5ex]{2pt}{0.4pt}}}\ULon} 
\pgfset{my bar/.tip={Bar[width=+0pt 20.2 0.89,line cap=round]}}
\tikzcdset{my mapsto/.code={\pgfsetarrows{my bar-tikzcd to}},
arrow style=tikz,
diagrams={>={Straight Barb[scale=0.7]}}}

\DeclareMathOperator{\Dom}{Dom}
\DeclareMathOperator{\Span}{Span}

\DeclareMathOperator{\CCl}{\mathbbm{C}l}

\renewcommand{\emph}{\textbf}

\title{\bf Isometric Spectral Subtriples\footnote{Corresponding author: W.Sucpikarnon.} 
\footnote{
This article may be downloaded for personal use only. Any other use requires prior permission of the author and AIP Publishing. This article appeared in \textit{Journal of Mathematical Physics 65(4):043504 (2024)}
and may be found at \href{https://doi.org/10.1063/5.0179837}{https://doi.org/10.1063/5.0179837}.
}
\footnote{
The present preprint version has been reformatted for arXiv purposes (authors are here listed in aphabetical order).}
}

\author{\normalsize  
\orcidicon{0000-0002-1387-9283} 
Paolo Bertozzini $^a$  
\ 
\orcidicon{0009-0007-3734-4180}
Wanchalerm Sucpikarnon $^b{}^*$ 
\ 
\orcidicon{0000-0003-2697-0315}
Apimook Watcharangkool $^c$ 
\\  
\\ 
\normalsize $^{a}$ \textit{Department of Mathematics and Statistics, Faculty of Science and Technology,}
\\
\normalsize \textit{Thammasat University, Pathumthani 12121, Thailand}
\\
\normalsize e-mail: $^a$ \texttt{paolo.th@gmail.com}  
\\ 
\normalsize $^b$ \textit{Department of Mathematics, Faculty of Science,} 
\\
\normalsize \textit{Chiang Mai University, Chiang Mai 50200, Thailand}
\\
\normalsize e-mail: $^b$ \texttt{wanchalerm.s@cmu.ac.th} 
\\ 
\normalsize $^c$ \textit{National Astronomical Research Institute of Thailand, Chiang Mai, 50180 Thailand}
\\
\normalsize e-mail: $^c$ \texttt{apimook@narit.or.th}
}

\date{\normalsize 24 April 2024}

\begin{document}

\maketitle

\begin{abstract}
    We investigate the notion of subsystem in the framework of spectral triple as a generalized notion of noncommutative submanifold. In the case of manifolds, we consider several conditions on Dirac operators which turn embedded submanifolds into isometric submanifolds. We then suggest a definition of spectral subtriple based on the notion of submanifold algebra and the already existing notions of Riemannian, isometric, and totally geodesic morphisms. We have shown that our definitions work at least in some relevant almost commutative examples.    
\end{abstract}

\section{Introduction}

 In a seemingly possible approach to quantum gravity one may construct geometrical objects that also capture common features of quantum theory. In doing so, one needs to consider geometric notions that extend beyond 
 the realm of 
 classical manifolds, for example, 
 manifolds with ``noncommutative coordinates''. There are many mathematical ways of realizing such a notion e.g. \cite{JMadore_1992, Madore-book}, \cite{Aschieri_2005, Aschieri-book}, \cite{Majid2000, Beggs-Majid-book}, 
 among these alternatives, the notion of spectral triple \cite{ConnesBook} aims at a vast generalization of Gel'fand-Na\u\i mark duality. Spectral triples can describe significant examples of  
 noncommutative manifolds, such as noncommutative tori, almost commutative manifolds etc.~\cite{Halleffect, AlmostNCG}.  
 These objects possess properties very similar to those of classical manifolds, such as Euler numbers, characteristic classes, index theory, etc. 
 However, several of the features present in the classical differential geometric level, are not yet fully available, for example, suitable notions of morphism, subobject and quotient object for spectral triples, are still under active investigation.

\medskip

In usual commutative differential geometry, there are different notions of submanifold corresponding to several choices of classes of differentiable maps with different ``rigidity'' (smooth, metric preserving, geodesic preserving). In some cases, such as spinorial Riemannian manifolds, even the notion of morphism is not clear: one does not have a ``natural'' way to lift maps between manifolds to maps between spinor bundles. In noncommutative geometry the problems are further complicated by the fact that the notion of ``quantum subsystem'' is not yet well-defined, even the choice of a category of C*-algebras presents us with several alternatives ($*$-homomorphisms, completely positive maps, bimodules of certain types). 

\medskip

Notions of subobject and quotient object can be deduced from a definition of morphism in a category. Some naive categories of spectral triples were proposed in \cite{Paolo2005, BCL11, Bertozzini_2012}, leaving the appropriate choice of morphism as an open question, nevertheless, some application of this approach in AF algebras can be found in the literature \cite{Flo-Gho, Contiquantum}. 
A much more technically sophisticated definition of morphism, based on Kasparov KK-theory, utilises suitable bimodule correspondences (equipped with extra smoothness and connection structures) \cite{Mesland2012, Mesland2014} and is at present the leading candidate for the discussion of morphisms between noncommutative spectral triples. Quotient spectral triples and, more recently, spectral subtriples have also been investigated along these lines of approach \cite{Kaad-Suij, Suij-Verh, Suij-Verh2}. 

\medskip 

The development of Connes' noncommutative differential geometry (based on Dirac operators and quantum calculus of forms) is not parallel with that of classical differential geometry and hence, within the formalism of spectral triples, most of the instruments readily available in usual differential geometry, e.g.~vector fields, Lie derivatives, curvatures, holonomies, cannot be immediately utilised or do not have a straightforward counterpart yet. However, in the tradition of derivation-based differential calculus \cite{DV-I,DV-M-II,DV-M}, noncommutative differential submanifolds (and quotient manifolds) have been defined \cite{doi:10.1063/1.531522}, in analogy with immersions (and submersions) of manifolds, via submanifold (and quotient manifold) algebras, that have been recently further investigated in \cite{SIGMA16}. 

\medskip 

In the present work, we incorporate the notion of submanifold algebra \cite{doi:10.1063/1.531522, SIGMA16} in the context of almost commutative spectral triples and, adopting some ideas from \cite{Paolo2005}, we put forward several tentative definitions of isometric spectral subtriples in the attempt to abstract the differential geometric features differentiating between global metric isometries and local totally geodesic or Riemannian isometries. In order to avoid the well-known problems with the definition of pull-back of spinor fields, we utilise the framework of Riemannian spectral triples \cite{LORD2012} motivated by the properties of the Hodge de Rham Dirac operator. In view of almost commutative examples we also suggest a modified definition of metric Connes' isometry exploiting a factorisation of the Grassmann algebra induced by the orthogonal decomposition of tanget space. Then we explore some of the relations among the different types of subtriples introduced.

\bigskip 

The paper is organized as follows. In section \ref{sub}, we first introduce the elementary definitions of naive spectral triple as strictly necessary in this paper (hence avoiding most of the technical axioms related to real structure and reconstruction theorems); then we give  the general definition of Riemannian, 
isometric embeddings  
between Riemannian manifolds together with their more restrictive cases of totally geodesic (respectively metric) isometric embeddings; 
following \cite{doi:10.1063/1.531522,SIGMA16}, we proceed to introduce submanifold algebras; 
and finally we describe the Hodge de Rham spectral triples originating from a Riemannian oriented compact manifold. 

\medskip 

In the main section \ref{Mor} we attempt to merge the submanifold algebra condition (specifying the ``smoothness/regularity'' of the morphism) with some naive ``metric'' notions of morphism of spectral triples in order to define isometric spectral subtriples and discuss some of the most immediate consequences of such definitions. 

\medskip 

In the final section \ref{Iso} we further investigate some of the relations among the different types of isometric subtriples previously introduced.  

\bigskip 

The approach currently taken in the present work has some quite strong limitations (that, of course, we plan to overcome in the near future): 
\begin{itemize}
\item 
the notions of spectral subtriples here considered are, for now, all based on the existence of unital $*$-homomorphisms between C*-algebras; although this (as shown by the examples provided) is perfectly fine in commutative and almost commutative cases, it becomes a too strong requirement when spectral triples over simple noncommutative algebras are considered; the utilisation of Hilbert bimodules (as in \cite{Mesland2014} and \cite{Suij-Verh2}), identifying the suitable classes of bimodules responsible for the alternative isometric subtriple conditions, should resolve this problem; 
\item 
in a perfectly similar way, derivations, as utilised now in derivation-based differential calculus and specifically here in the definition of submanifold algebra, are way too restrictive for fully noncommutative algebras: derivations of a noncommutative algebra do not form a bimodule over the algebra; an appropriate noncommutative generalization of the bimodule of vector fields on a manifold will have to be based on right/left derivations (see \cite{borowiec} and also \cite{chatchai}) and hence we expect that also the notion of submanifold algebra and smoothness of morphisms between spectral triples, will have to undergo similar reformulations (that we will pursue elsewhere); 
\item 
our definitions of isometric spectral subtriples are, for now, crucially motivated by the Hodge de Rham Riemannian case and hence they should apply to the Riemannian spectral triples of \cite{LORD2012}; we did not examine here which modification might take place in the case of the more usual setting of spinorial (Atiyah Singer) Dirac spectral triples;  

\item even with the above-mentioned limitations, most of the discussion provided here will be concerned with, and supported by, examples where there is a ``bounded'' Hilbert space map (usually a co-isometry), spatially implementing  a ``restriction'' homomorphism (instead of a conditional expectation), and hence dealing with situations that in Riemannian geometry correspond essentially to embeddings of families of connected components, without dimension change (but see remark~\ref{rem: Rsub} and the itemized points after definition~\ref{def: iso} for some preview on a more general setting); 

\item 
in this work we have been able to explore only some of the most elementary relations between alternative formalisations of naive isometric embedding of spectral triples; much more remains to be said, in particular general theorems of reconstruction for (isometric) embeddings and other morphisms of commutative spectral triples are still missing. 
\end{itemize}

\section{Spectral Triple of an Isometric Submanifold}
\label{sub}

We quickly recall the basic notion of naive spectral triple in order to set up our notation in view of the discussion of morphisms and subtriples. 

\begin{Def} \label{def: st}
A (naive) {\bf unital spectral triple} $(\mathcal{A},\mathcal{H},\mathcal{D})$ consists of: 
\begin{itemize}
\item 
a complex Hilbert space $\mathcal{H}$,
\item 
a unital complex C*-algebra $\mathcal{A}$ faithfully represented $\mathcal{A}\xrightarrow{\pi}\mathcal{B}(\mathcal{H})$ by bound\-ed operators acting on $\mathcal{H}$, 
\item
a (possibly unbounded) {\bf Dirac operator} $\mathcal{D}:\Dom(\mathcal{D})\to\mathcal{H}$ such that: 
\begin{itemize}
\item 
$\mathcal{H}^1:=\Dom(\mathcal{D})\subset\mathcal{H}$ is dense in $\mathcal{H}$,  
\item 
$\mathcal{D}$ is self-adjoint: $\mathcal{D}^*=\mathcal{D}$, 
\item 
for all $x\in \mathcal{A}$, $\pi(x) (\mathcal{H}^1)\subset \mathcal{H}^1$, 
\item 
the unital $*$-subalgebra $\mathcal{A}^1:=\{x\in\mathcal{A} \ | \ [\mathcal{D}, \pi(x)]\in \mathcal{B}(\mathcal{H}^1;\mathcal{H}) \}$ is norm dense in $\mathcal{A}$.
\end{itemize}
\end{itemize}
Given a spectral triple $(\mathcal{A},\mathcal{H},\mathcal{D})$, its associated {\bf complex Clifford algebra} is:
\begin{equation*}
\CCl_{\mathcal{D}}(\mathcal{A}):=\overline{\Span (\{\pi(x_0)[\mathcal{D},\pi(x_1)]\cdots [\mathcal{D},\pi(x_n)] \ | \ n\in\mathbbm{N}, \ x_0,\dots,x_n\in\mathcal{A}^1\})}. 
\end{equation*}
The {\bf Dirac flux} of the spectral triple is defined as the one-paramenter group: 
\begin{equation*}
\mathbbm{R} \ni t\mapsto \alpha^D_t(T):=\exp(iDt) T \exp(-iDt) \in \mathcal{B}(\mathcal{H}), \quad \forall T\in \mathcal{B}(\mathcal{H}), 
\end{equation*}
the unital $*$-subalgebra of $\alpha^D$-smooth elements of $\mathcal{A}$ is defined as: 
\begin{equation*}
\mathcal{A}^\infty:=\{x\in\mathcal{A} \ | \ t\mapsto\alpha^D_t(\pi(x)) \ \text{is $C^\infty(\mathbbm{R};\mathcal{B}(\mathcal{H}))$} \}.
\end{equation*}
The spectral triple is said to be {\bf smooth} if $\mathcal{A}^\infty$ is norm dense in $\mathcal{A}$ and $\mathcal{H}^\infty$ is dense in $\mathcal{H}$, where: 
\begin{equation*}
\mathcal{H}^\infty:=\bigcap_{n=1}^\infty\Dom(\mathcal{D}^n). 
\end{equation*}
\end{Def}
Notice that the smoothness of the spectral triple already entails the condition that $\mathcal{A}^1$ is norm dense in $\mathcal{A}$ as well as the density of $\Dom(\mathcal{D})$ in $\mathcal{H}$. 

\medskip 

Spectral triples might (and will) satisfy many more of the several additional requirements necessary for reconstruction theorems in \cite{Connes,LORD2012}, in the previous definition we introduced only a minimal set of axioms that are absolutely necessary for our discussion. 

\medskip 

It is important to stress that, in the definition above, we did not ask for the irreducibility of the representation $\mathcal{A}\xrightarrow{\pi}\mathcal{B}(\mathcal{H})$, since later in subsection \ref{sec: HdR} we will need to consider Hodge de Rham spectral triples, where the Clifford algebra $\CCl_{\mathcal{D}}(\mathcal{A})$ has a huge anti-isomorphic commutant $\CCl_{\mathcal{D}}(\mathcal{A})'$, since $\mathcal{H}$ is equipped with a cyclic separating vector, for the von Neumann Clifford algebra $\CCl_{\mathcal{D}}(\mathcal{A})''$. 

\medskip 

The somehow idiosyncratic notion of ``Dirac flux'' is essential in order to identify a ``smooth'' subalgebra $\mathcal{A}^\infty$ of $\mathcal{A}$ on which the submanifold algebra condition in \ref{sec: sub} will have to be applied later on (see definition \ref{def: iso}). 

\medskip 

For a C*-algebra $\mathcal{A}$, let us denote by $\mathscr{S}(\mathcal{A})$ the set of its states and by $\mathscr{P}(\mathcal{A})$ the subset of its pure states. Any naive spectral triple $(\mathcal{A},\mathcal{H}, \mathcal{D})$ induces on the set $\mathscr{S}(\mathcal{A})$ of ($\mathscr{P}({\mathcal{A}})$ pure) states the following {\bf Connes' distance}:\footnote{
This is actually a pseudo-metric since $d_{\mathcal{D}}(\rho,\sigma)$ can attain the value $+\infty$.} 
\begin{equation}\label{eq: Connes-d}
d_{\mathcal{D}}(\rho,\sigma):=\sup\{|\rho(x)-\sigma(x)| \ | \ \|[\mathcal{D},\pi(x)]\| \leq 1 \}, \quad \forall \rho,\sigma\in\mathscr{S}(\mathcal{A}). 
\end{equation}

If $(\mathcal{A}_1,\mathcal{H}_1,\mathcal{D}_1)$ and $(\mathcal{A}_2,\mathcal{H}_2,\mathcal{D}_2)$ are two unital commutative spectral triples satisfying the axioms of the reconstruction theorem \cite{Connes}, then one can construct associated Riemannian compact spin manifolds $(N,g)$ and $(M,g')$, respectively. 

The question is what are the conditions needed for the spectral triples so that $N$ is a submanifold of $M$, and how to generalize these conditions to noncommutative cases. Let us first investigate the relation between the canonical spectral triples of the two compact orientable Riemannian manifolds when $N$ is an $n$-dimensional submanifold immersed into an $m$-dimensional manifold $M$. 
In this paper we will be interested in isometric immersions. There are several different types of isometric immersion between compact orientable Riemannian manifolds, depending on the level of ``rigidigy'' imposed. The most general one is the following Riemannian (local) isometry. 

\begin{Def} 
A function $f:N \rightarrow M$ is a Riemannian isometric immersion if, for all local charts $(U,\psi_N)$ of $N$ and  $(V,\psi_M)$ of $M$, the map 
\begin{equation*}
\tilde{f}:=\psi_M\circ f \circ \psi_N^{-1}
\end{equation*}is the solution of the nonlinear partial differential equations
\begin{equation}
(\tilde{f}_*)^a_{
~i}(\tilde{f}_*)^b_{~j} g'_{ab}(\tilde{f}(x))=g_{ij}(x), \label{isom}
\end{equation}
for all $x=(x^1,...,x^n)\in \psi_N(U)$, where $(\tilde{f}_*)^a_{~i}=\partial \tilde{f}^a/\partial x^i$, $a,b$ and $i,j$ are indices of local coordinates on $M$ and $N$, of which metric tensors are $g'$ and $g$, respectively, where we use Einstein summation convention to simplify the expression. 
\end{Def}

Equivalently, a Riemannian (local) isometry is an immersion $f:N\to M$ such that $f^*(g')=g$. 
Notice also that the above condition for Riemannian (local) isometric immersion is equivalent to requiring that the differential (tangent map) $T_pN\xrightarrow{f_*} T_{f(p)} M$ is fibrewise a linear isometry, for all $p\in N$, and hence $TN\xrightarrow{f_*}TM$ is a morphism of Hermitian vector bundles. 

\medskip 

Note that, by definition, $f:N\to M$ is a local isometry, but this does not always imply that $f$ is preserving the geodesic distances. 
However, the local isometry condition does provide the following properties \cite{Peter} (that will be revisited in the noncommutative case in section \ref{Iso}):
\begin{description}
    \item{$i)$}~~ $f$ maps geodesics of $N$ to geodesics of $f(N)$ (that do not always coincide with geodesics of $M$), 
    \item{$ii)$}~ $f$ is geodesic distance decreasing,  
    \item{$iii)$} if $f$ is also a bijection, it is geodesic distance preserving.
\end{description}

In addition, a Riemannian isometric immersion $f$ induces the $f$-pull-back vector bundle $f^*(TM)$ (of the tangent bundle $TM$ of $M$), which is a bundle over $N$ whose fibre at $p\in N$ is $T_{f(p)}M$. Such pull-back bundle admits (see \cite[section 1.1]{Subman-intro}) the orthogonal decomposition 
\begin{align}
\label{Decom}
    f^*(TM)&=  f^*(f_*(TN)) \oplus f^*((f_*(TN))^\perp) \nonumber\\
         &\simeq TN\oplus\mathcal{N}_fN,
\end{align}
where the tangent map $f_*: TN\to TM$ is a linear isometric isomorphism onto its image and we denote by  $\mathcal{N}_fN:=f^*(f_*TN)^\perp$ the normal bundle. 

\begin{Def}
We will say that a Riemannian (local) isometry $N\xrightarrow{f}M$ is a: 
\begin{itemize}
\item 
{\bf totally geodesic (local) isometry} if the restriction to $f(N)$ of the Levi-Civita connection $\nabla^M$ of $M$ has a diagonal form with respect to the orthogonal decomposition \eqref{Decom};
\item 
{\bf (smooth global) metric isometry} if it is preserving the geodesic distance betwen the two manifolds. 
\end{itemize}
\end{Def}
The Myers Steenrod theorem (see for example \cite[theorem 5.6.15]{Peter}) assures that surjective global isometries are always smooth and hence Riemannian isometries. Since the theorem requires surjectivity, here we assume the smootheness a priori in the definition of global metric isometry. 

Totally geodesic (local) isometries map geodesics curves of $N$ onto geodesic curves of $M$, but does not necessarily preserve the geodesic distances (as can be easily seen mapping a circle into a geodesic helix winding more than one time around a flat 2-torus). 

Riemannian isometries are not necessarily totally geodesic, as can be seen mapping a circle onto a non-maximal circle on a 2-sphere.  

In practice we have the following strict implications that somehow we expect to survive in the noncommutative setting: 
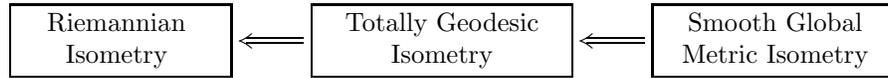
\begin{equation*}
\xymatrix{
\boxed{\text{\begin{minipage}{2.6cm} \begin{center}
Riemannian Isometry 
\end{center}
\end{minipage}
}} & \ar@{=>}[l] \boxed{\text{
\begin{minipage}{3.0cm} \begin{center}
Totally Geodesic \\ Isometry
\end{center}
\end{minipage}
}} & \ar@{=>}[l] \boxed{\text{
\begin{minipage}{2.8cm}\begin{center}
Smooth Global \\ Metric Isometry
\end{center}
\end{minipage}
}}
}
\end{equation*}

\subsection{Submanifold Algebra} \label{sec: sub}

Let $\varphi: \mathcal{A} \rightarrow \mathcal{B}$ be a surjective homomorphism between associative algebras over a field $\mathbbm{K}$. 
The set ${\rm Der}(\mathcal{A})$ denotes the collection of derivations $D$, which are $\mathbbm{K}$-linear maps satisfying 
\begin{equation*}
    D(ab)=(Da)b+a(Db), ~~a,b \in \mathcal{A}.
\end{equation*}
Given the ideal $\mathcal{I}={\rm ker}~\varphi$ we can define a set 
\begin{equation*}
{\rm Der}_\varphi(\mathcal{A}):=\{D\in {\rm Der}(\mathcal{A}) \ | \ Da \in \mathcal{I}, \ \forall a\in \mathcal{I} \},
\end{equation*}
and then, since $\mathcal{A} / \mathcal{I} \simeq \mathcal{B}$, as in \cite{SIGMA16}, we have the linear map 
\begin{equation*}
\varphi_*:{\rm Der}_\varphi(\mathcal{A})\rightarrow {\rm Der}(\mathcal{B}), 
\end{equation*}
associating $D\in {\rm Der}_\varphi(\mathcal{A})$ to the derivation of $\mathcal{B}$ corresponding to the derivation $\hat{D}\in{\rm Der}(\mathcal{A} / \mathcal{I})$, where $\hat{D} (a+\mathcal{I}):=Da +\mathcal{I}$, is well-defined for all $a+\mathcal{I}\in \mathcal{A} / \mathcal{I}$. 

\begin{Def} (see \cite{SIGMA16})  
$\mathcal{B}$ is a submanifold algebra of $\mathcal{A}$ if $\varphi_*$ is surjective.
\end{Def}
Suppose $f$ is an isometric immersion, since $N$ is compact, $f$ is also an embedding with closed image. 
By Theorem 11 of \cite{SIGMA16}, one has a surjective homomorphism $\varphi:C^\infty(M)\rightarrow C^\infty(N)$ which is equal to $f^*$ \cite{doi:10.1063/1.531522}. 
The map $\varphi=f^*$ induces the following short exact sequence 
\begin{equation}
    0 \rightarrow \mathcal{I} \rightarrow C^\infty(M) \xrightarrow{\varphi} C^\infty(N) \rightarrow 0, \label{SplitEx}
\end{equation}
where $\mathcal{I}$ is the ideal of smooth functions vanishing on $f(N)$. 

In addition, the map $\varphi_*:{\rm Der}_\varphi(C^\infty(M)) \rightarrow {\rm Der}(C^\infty(N)) $ is also surjective, where
\begin{equation*}
    {\rm Der}_\varphi(C^\infty(M))=\{X\in {\rm Der}(C^\infty(M))~|~ Xg \in \mathcal{I}, \ \forall g\in \mathcal{I}\}.
\end{equation*}
Hence, the algebra $C^\infty(M)/\mathcal{I}\cong C^\infty(N)$ is a submanifold algebra of $C^\infty(M)$. 

In the following, with some abuse of notation, for uniformity, we will write 
\begin{equation*}
\Gamma(TM_\varphi):= \text{Der}_\varphi(C^\infty(M)),  
\end{equation*}
where $\Gamma(-)$ for us denotes the family of smooth sections of a ``smooth bundle''.\footnote{
Notice that the ``spectrum'' $TM_\varphi$ of the (non-projective) $C^\infty(M)$-module $\text{Der}_\varphi(C^\infty(M))$ is not a locally trivial vector bundle over $M$ and hence Serre-Swan theorem (see for example~\cite[section~2.3]{GBVF}) cannot be invoked for its construction: its fibers are $(TM_\varphi)_p \simeq T_pM$ when $p\in M-f(N)$ and $(TM_\varphi)_{f(q)}\simeq T_qN$ whenever $p=f(q)\in f(N)$; we do not enter here into the technicalities involved in the definition of topology and smooth structure on $T_\varphi M$. 
} 

To see the effect of the derivation $X\in {\rm Der}_\varphi(C^\infty(M))$ in local coordinates, let $g\in \mathcal{I}$ and suppose $\{y^1, \ldots, y^n\}$ and $\{x^1, \ldots , x^m\}$ are local coordinates at $p\in N$, and $f(p)\in M$, respectively. Since $f$ is an embedding, one can choose coordinates such that the first $n$ coordinates in $f(M)$ are functions of $y^i$'s, while the rest do not depend on $y^i$'s. Consider 
\begin{align}
    0=&Xg|_{f(p)}=\left. X^i(f(p))\frac{\partial g}{\partial x^i}\right|_{f(p)}\nonumber\\
    =& \left. X^{n+\alpha}(f(p))\frac{\partial g}{\partial x^{n+\alpha}}\right|_{f(p)}, \label{TMpi}
\end{align}
where $\alpha=1, \ldots , m-n$. The derivatives $\partial g/\partial x^i,~i=1,\ldots, n$, vanish, since $g=0$ on $f(N)$. For $i=n+\alpha$, with $\alpha=1,\dots, m-n$, the derivatives do not necessarily vanish, the condition $Xg|_{f(p)}=0$ implies $X^{n+\alpha}(f(p))=0$.

\begin{Lem}
\label{Tangent}
Let $(M,g)$ be a compact Riemannian manifold, and $C^\infty(N)$ be a submanifold algebra of $C^\infty(M)$. 
The module of smooth sections of the induced tangent bundle $f^*(TM)$ admits the decomposition
\begin{equation}
  \Gamma(f^*(TM)) = \varphi_*(\Gamma(TM_\varphi)) \oplus \Gamma(\mathcal{N}_fN)\simeq \Gamma(TN)\oplus \Gamma(\mathcal{N}_fN),  \label{Dsum}
\end{equation} 
where $f$ is a proper embedding, and $\mathcal{N}_fN$ is the normal bundle associated to $f$. 
\end{Lem}
\begin{proof}
By \cite[theorem~11]{SIGMA16}, a submanifold algebra gives rise to a proper embedding \hbox{$N\xrightarrow{f} M$.} The module of section of the induced tangent bundle $f^*(TM)$ admits the decomposition \eqref{Decom}:
\begin{equation*}
    \Gamma(f^*(TM))= 
     \Gamma(f^*(f_*(TN))\oplus\Gamma(\mathcal{N}_fN)
    {\simeq \Gamma(TN)}\oplus\Gamma(\mathcal{N}_fN).
\end{equation*}
It is clear from \eqref{TMpi} that the element of $\Gamma(f^*TM_\varphi)$ does not belong to the normal bundle, i.e. $\Gamma(f^*TM_\varphi) \subset \Gamma(TN)$ (up to canonical isomorphism). 
Let $p \in N$ and employ the coordinates used in \eqref{TMpi}. 
Let $X \in \Gamma(TN)$, and $Y=f_*X \in \Gamma(f_*TN)$. Suppose $g\in \mathcal{I}$, then 
\begin{align*}
    Y|_{f(p)} g =&~ (f_*X)_{f(p)}g \\ 
              =&~ \left.\frac{\partial y^J}{\partial x^I}X^{I}\frac{\partial g}{\partial y^J}\right|_{f(p)} \\
              =&~0,
\end{align*}
where $I,J=1, \ldots , n$. Recall that $g$ vanishes on $f(N)$ and hence $X\in \Gamma(f^*TM_\varphi)$ which makes $\Gamma(f^*TM_\varphi) \simeq\Gamma(TN) = \varphi_*\Gamma(TM_\varphi$).  
\end{proof}
Once the direct sum decomposition~\eqref{Dsum} of pre-Hilbertian \hbox{C*-mod}\-ules over $C^\infty(N)$ is in place, the Fermionic second quantization and complexification functors, provide the following canonical factorizations of the Grassmann algebras and of the Clifford algebras of the $f$-pull-back of $TM$:  
\begin{align*}
\Omega_{\mathbbm{C}} (\Gamma(f^*(TM)))&= \Omega_{\mathbbm{C}}(\varphi_*(\Gamma(TM_\varphi))) \otimes_{C^\infty(N)} 
\Omega_{\mathbbm{C}}(\Gamma(\mathcal{N}_fN))
\\
&\simeq\Omega_{\mathbbm{C}}(\Gamma(TN))\otimes_{C^\infty(N)} \Omega_{\mathbbm{C}}(\Gamma(\mathcal{N}_fN)), 
\\ \\
\CCl(\Gamma(f^*(TM)))&=\CCl(\varphi_*(\Gamma(TM_\varphi)))\otimes_{C^\infty(N)}\CCl(\Gamma(\mathcal{N}_fN))
\\
&\simeq \CCl(\Gamma(TN))\otimes_{C^\infty(N)}\CCl(\Gamma(\mathcal{N}_fN)). 
\end{align*}

\medskip 

Whenever the manifold $M$ is equipped with a Riemannian metric tensor $g_M$, the inner product defines an isometric isomorphism $TM \xrightarrow{\Lambda_{g_M}} T^*M$ between tangent and cotangent bundles, therefore one has a ``dual version'' of the previous tensorization of Clifford algebras, that will be used in the next section: 
\begin{align*}
C^\infty(N)\otimes_{C^\infty(M)}\CCl(\Gamma(T^*M))
&\simeq 
\CCl(\Gamma(f^*(T^*M))) 
\\
&
\simeq \CCl(\Gamma(T^*N))\otimes_{C^\infty(N)}\CCl(\Gamma(\mathcal{N}^*_fN)). 
\end{align*}

\subsection{Hodge de Rham Spectral Triples} \label{sec: HdR}

A very special class of naive spectral triples (definition~\ref{def: st}) originates from compact orientable Riemannian manifolds $(M,g)$ and has been axiomatically studied in \cite{LORD2012}, where conditions for a full reconstruction theorem have been produced (see also \cite{FGR}, for an earlier suggestion in this direction and \cite{Twisted} for further investigation): 
\begin{Def}
The \textbf{Hodge de Rham spectral triple} $(\mathcal{A}_M,\mathcal{H}_M,\mathcal{D}_M)$ of the compact orientable Riemannian manifold $(M,g_M)$ consists of: 
\begin{itemize}
\item 
$\mathcal{A}_M:= C(M)$, the unital C*-algebra of continuous functions on $M$,
\item 
$\mathcal{H}_M:=L^2(\Omega_{\mathbbm{C}}(M))$, the Hilbert space obtained by completion of the comp\-lex\-i\-fied Grassmann algebra 
$\Omega_{\mathbbm{C}}(M):=\Omega_{\mathbbm{C}}(\Gamma(T^*M))$ of smooth differential forms on $M$ under the $\mathbbm{C}$-valued inner product $\langle \rho | \sigma\rangle_{\mathbbm{C}}:=\int_M \overline{(\star \rho)}\wedge \sigma$, for $\rho,\sigma\in \Omega_{\mathbbm{C}}(M)$, where $\star\rho$ denotes the Hodge dual of $\rho$,  

\item 
$\mathcal{D}_M$, the Hodge de Rham Dirac operator of $M$, that is the closure of the densely defined operator $d_M+d^*_M$, where $d_M$ denotes the, densely defined, exterior differential of $M$ and $d_M^*$ its Hilbert space adjoint on $\mathcal{H}_M$.   
\end{itemize} 
\end{Def}

The inner product $g_M$ defines an isometric isomorphism $TM \rightarrow T^*M$. 
The complexified Clifford algebra $\CCl(M):=\CCl(\Gamma(T^*M))$ of a Riemannian manifold acts on the $C^\infty(M)$-pre-Hilbert C*-module of complexified differential forms $\Omega_{\mathbbm{C}}(M)$ by endomorphisms $\CCl(M)\xrightarrow{c_M}\text{End}_{C^\infty(M)}(\Omega_{\mathbbm{C}}(M))$ 
with the action generated by creation and annihilation operators
\begin{equation*}
c_M(v)\cdot\xi=v\wedge\xi+v\mathbin{\lrcorner}\xi, 
\quad \forall v\in \mathbbm{C}\otimes_{\mathbbm{R}} \Gamma(T^* M), \ \xi\in\Omega_{\mathbbm{C}}(M), 
\end{equation*}
which satisfy the Clifford relations
\begin{equation*}
c_M(v)^2\cdot\xi=v\wedge v\mathbin{\lrcorner}\xi+v\mathbin{\lrcorner}v\wedge\xi=\|v\|^2\xi, 
\end{equation*}
and for the commutators with the Hodge de Rham Dirac operator (see for example~\cite[proposition 3.8]{BGV}) we have $[\mathcal{D}_M,x]=c_M(d_M x)$, for all $x\in C^\infty(M)$. 
Using the notation in definition~\ref{def: st}, we have $\CCl_{\mathcal{D}_M}(\mathcal{A}_M)=\overline{c_M(\CCl(M)}\subset \mathcal{B}(\mathcal{H}_M)$. 

\begin{Thm}
\label{Lhimma}
Suppose $N \xrightarrow{f} M$ is a smooth function between compact orientable Riemannian manifolds with metric tensors $g_M$ and $g_N$ defined on $M$ and $N$, respectively. Let 
$(\mathcal{A}_M,\mathcal{H}_M,\mathcal{D}_M)$ and $(\mathcal{A}_N,\mathcal{H}_N,\mathcal{D}_N)$
be the two Hodge de Rham spectral triples of $M$ and $N$, then the following are equivalent:
\begin{enumerate}
    \item[i)] $f$ is an open Riemannian isometric embedding.
    \item[ii)] $C^\infty(N)$ is a submanifold algebra of $C^\infty(M)$ such that the exact sequence \ref{SplitEx} splits, and hence $C^\infty(M)\simeq C^\infty(f(N))\oplus C^\infty(M-f(N))$, and we have 
    \begin{equation*}
 \|\, [ \mathcal{D}_N,\varphi(h)]\, \|_{\mathcal{B}(\mathcal{H}_N)}
=    \|\, [  \mathcal{D}_M,h]\, \|_{\mathcal{B}(\mathcal{H}_M)}, 
 \end{equation*}
 for all $h\in C^\infty(f(N))\oplus K(M-f(N))$, where $\varphi:=f^*$ and we denote the set of constant functions on the complement of $f(N)$ by $K(M-f(N)):=\mathbbm{C}\cdot  1_{M-f(N)}$.   
\end{enumerate}
\end{Thm}
\begin{proof}
Since the manifolds are compact Hausdorff, any smooth immersion is an embedding and hence the image is compact and closed. It follows that $f(N)$ is a disjoint union of connected components of $M$. The submanifold algebra condition is already valid, as for any embedded smooth submanifold. 

Since $\|[\mathcal{D}_M,h]\|=\|\text{grad}_M(h)\|$ and $h=g\oplus k\in C^\infty(f(N))\oplus K(M-f(N))$, one obtains that 
$\text{grad}_M(h)=\text{grad}_{f(N)}g=\text{grad}_Nf^*g=\text{grad}_N(\varphi(h))$.
Hence the norm equality of commutators holds. 

\medskip 

Since $\CCl_{\mathcal{D}_M}(C(M)) = \overline{c_M(\CCl(M))}$ as algebra of operators acting on $\mathcal{H}_M$, and since under this isomorphism $[\mathcal{D}_M,h]$ corresponds to the Clifford multiplication by $d_Mh\in\CCl(M)$, under the isometric commutators condition,  $\varphi$ extends to a continuous homomorphism between the complexified Clifford algebras of the Riemannian manifolds and hence it is a Riemannian isometry. If, by contradition, this Riemannian embedding is not open, we can construct a smooth function $h$ on $M$ with a given restriction to $f(N)$ with at least one point $p\in M-f(N)$ where $\|(\text{grad}_M\  h) (p)\|>\max_{q\in f(N)} \|(\text{grad}_{f(N)} \ h)(q)\|$ that contradicts the equality of the norms of commutators. 
\end{proof}
\begin{Rem}\label{rem: Rsub}
The previous result might seem somehow disappointing: it will actually imply that ``isometric spectral subtriples'' satisfying our point~d.~in definition~\ref{def: iso} in the commutative case (embedding of compact orientable Riemannian manifolds) exist only in the very trivial cases of Riemannian totally geodesic embeddings consisting of inclusions of families of connected components of a given Riemannian orientable manifold, in which cases the Hilbert space decomposes as the direct sum: $L^2(\Omega_{\mathbbm{C}}(M))\simeq L^2(\Omega_{\mathbbm{C}}(f(N)))\oplus L^2(\Omega_{\mathbbm{C}}(M-f(N)))$. 

\medskip 

The situation is anyway not so bad and one can achieve much more. Since the technical details involved are quite serious, we plan to examine the full implications of such lines in a separate study (that will necessarily require significant changes to the naive type of morphism of spectral triples used here in section~\ref{Mor}), and we only put-forward some general arguments that justify our position. 

\medskip 

Any compact Riemannian embedding $N\xrightarrow{f}M$ admits (see for example~\cite[section~6.2]{CS}) an open $r$-tubular neighborhood  $f(N)\subset T_r\subset M$ of $f(N)$ on which a Riemannian submersion $T_r\xrightarrow{\pi_r}f(N)$ exists such that $\pi_r|_{f(N)}$ is the identity.\footnote{
Notice that, in general, $T_r\xrightarrow{F_r} f(N)$ is not necessarily a trivial bundle. 
}
Denoting by $T_r\xrightarrow{F_r}N$ the smooth map $F_r:=f^{-1}\circ\pi_r$, we see that $F_r\circ f=\textrm{Id}_N$. \\ 
Utilizing the co-area theorem~\cite{Nicolaescu-notes} in the case of the previous submersion, we obtain a bounded partial Fubini integration map $u_r:L^2(\Omega_{\mathbbm{C}}(T_r))\to L^2(\Omega_{\mathbbm{C}}(N))$ that generalizes the usual Hilbert projection in the case of product of manifolds. 
Consider the orientable Riemannian ``open'' submanifold $(T_r,g_M|_{T_r})$ and 
$(\mathcal{A}_{T_r},\mathcal{H}_{T_r},\mathcal{D}_{T_r}):=(C_b(T_r),L^2(\Omega_{\mathbbm{C}}(T_r)),\mathcal{D}_M|_{T_r})$, 
its ``non-compact'' Riemannian spectral triple, obtained by restriction to $T_r$, where $C_b(T_r)$ denotes the unital C*-algebra of bounded continuous functions on $T_r$. 

\medskip 

Using the fact that the tangent map $(F_r)_*$ of the retraction $F_r$ is a bundle co-isometry, we have a unital $*$-homomorphism $\CCl_{\mathcal{D}_N}(N)\xrightarrow{\phi_r}\CCl_{\mathcal{D}_{T_r}}(T_r)$, that induces on $\CCl_{\mathcal{D}_{T_r}}(T_r)$ the structure of a right-$\CCl_{\mathcal{D}_N}(N)$ module; 
furthermore the morphism $\phi_r$ has  
a left-inverse $\CCl_{\mathcal{D}_N}(N)\xleftarrow{\psi_r}\CCl_{\mathcal{D}_{T_r}}(T_r)$ that is a ``conditional expectation'' of $\CCl_{\mathcal{D}_{T_r}}(M)$ onto $\CCl_{\mathcal{D}_N}(N)$. 
As a consequence of Stinespring's theorem~\cite[theorem II.7.5.2]{Black}, there is a GNS-representation for conditional expectations that will provide us with a right-$\CCl_{\mathcal{D}_N}(N)$ C*-module $L^2(\psi_r)$ inducing the following tensor product decomposition for Hilbert \hbox{C*-modules}: 
$\mathcal{H}_{T_r}\simeq L^2(\psi_r)\otimes_{\CCl_{\mathcal{D}_N}(N)} \mathcal{H}_N$ that generalizes the usual tensor product of Hilbert spaces in the product case. 
The Hodge de Rham Dirac operators are expected to ``factorize'' only when the embedding is totally geodesic, still there is on $T_r$ a ``cutting projection'' situation that we will examine in further work. 

\medskip 

We warn in advance the reader that the pair $(f^*,u_r)$ will not be a naive morphism of spectral triples as we will define in the subsequent section~\ref{Mor} (since $u_r(x\cdot \xi)\neq f^*(x)\cdot u_r(\xi)$, for $x\in \mathcal{A}_M$ and $\xi\in \mathcal{H}_M$) and hence also the ``cutting projection'' will not provide an isometric subspectral triple according to point d.~in our definition~\ref{def: iso}. 
Again, one can overcome this problem (and eliminate as well the arbitrariness due to the choice of the parameter $r$) passing to the ``inductive limit for $r\to 0$'' of the spectral triples of the family of tubular neighborhoods $T_r$ (using the definition in~\cite{Flo-Gho}); the price to pay will be the need of a more sophisticated notion of (co-span) morphism of spectral triples (consisting of two combined processes of ``transport'' and ``correspondence'', along the lines already suggested in~\cite{F}) that we will necessarily examine elsewhere. 
\end{Rem}

We ``formalize'' the previous discussion recalling these definitions. 
\begin{Def}
A \textbf{submersion} is a smooth map $F:M\to N$, between manifolds, such that its tangent map $F_*: TM\to TN$ is fibrewise surjective. 
\\ 
A \textbf{Riemannian submersion} is a submersion $F:M\to N$, between Riemannian manifolds, such that its tangent map $F_*:TM\to TN$ is a fibrewise isometry when restricted to the subbundle $\ker(F_*)^\perp\subset TM$. 
\end{Def}

\begin{Thm} \label{th: Tr}
Let $f:N\to M$ be a Riemannian immersion that admits a \textbf{retraction} $F:M\to N$ (a smooth map such that $F\circ f(x)=x$, for all $x\in N$) that is a Riemannian submersion  
Then, with the same assumptions and notation as in theorem~\ref{Lhimma}, we have that:
$C^\infty(N)$ is a submanifold algebra of $C^\infty(M)$ and 
$ \|\, [ \mathcal{D}_N,f^*(h)]\, \|_{\mathcal{B}(\mathcal{H}_N)}    =    \|\, [  \mathcal{D}_M,F^*(f^*(h))]\, \|_{ \mathcal{B}(\mathcal{H}_M)}$, for all $h\in C^\infty(M)$.
\end{Thm}
\begin{proof}
The submanifold algebra condition is true, as usual, considering 
$\varphi:=f^*: C^\infty(M)\to C^\infty(N)$ the pull-back by the embedding $f$.

\medskip 

Since $F:M\to N$ is a Riemannian retraction of the Riemannian embedding $f:N\to M$, the tangent map $f_*: TN \to TM$ is a fibrewise isometry; the tangent map $F_*: TM\to TN$ is fibrewise a co-isometry, $F_*\circ f_*=I_{TN}$, we have an orthogonal decomposition of the tangent bundle $TM\simeq (\ker F_*) \oplus (\ker F_*)^\perp$ as a direct sum of ``vertical'' and ``horizontal'' subbundles, furthermore any vector field $X\in \Gamma(TN)$ admits a unique horizontal lift $\overline{X}\in\Gamma((\ker F_*)^\perp)$ such that $F_*(\overline{X})=X$.

\medskip 

For every $k\in C^\infty(N)$ 
we have that $\text{grad}_M (F^*k)\in (\ker F_*)^\perp$ since, for every vector field $Z\in \ker F_*$: 
\begin{align*}
g_M(\text{grad}_M (F^*k)), Z)
=
(d_M(F^*k))(Z)
= 
(F^*(d_Nk))(Z)
=d_Nk(F_*Z)=0.
\end{align*}

\medskip 

Furthermore, we have that $F_*(\text{grad}_M (F^*k))=\text{grad}_N k$. 
For every vector field $X\in \Gamma(TN)$, consider the unique lifted horizontal vector field $\overline{X}\in \Gamma(TM)$: 
\begin{align*}
g_N(F_*(\text{grad}_M &(F^*k)),X)
=g_M(\text{grad}_M (F^*k)), \overline{X})
=d_M(F^*k)(\overline{X})
\\
&=(F^*d_Nk)(\overline{X})
=d_Nk(F_*(\overline{X}))
=d_Nk(X)
=g_N(\text{grad}_Nk,X).
\end{align*}

\medskip 

Since $F_*: (\ker F_*)^\perp \to TN$ is a surjective isometry, we obtain that: 
\begin{equation*}
\|\text{grad}_M (F^*k)\|=\| \text{grad}_N k\|,
\end{equation*}
hence $\| [\mathcal{D}_M,F^*k] \|=\| [\mathcal{D}_N,k] \|$ and so, taking $k:=f^*(h)$ for $h\in C^\infty(M)$, we obtain 
\begin{equation*}
\|[\mathcal{D}_N,f^*(h)]\|_{\mathcal{B}(\mathcal{H}_N)}=\|[ \mathcal{D}_M , F^*(f^*(h))]\|_{\mathcal{B}(\mathcal{H}_M)}.
\end{equation*}
\end{proof}

Products of compact orientable Riemannian manifolds are special cases covered by the previous theorem~\ref{th: Tr}: in this case we can take a global tubular neighborhood $T_{\hat{r}}=N\times M$, with $\hat{r}:=\text{diam}(M)$, $f$ as the inclusion of $N$ in the product $M\times N$ and $F$ as the projection onto the second component. 

\medskip 

The following example of embedding is paradigmatic of theorem~\ref{th: Tr} and it is more general than the usual Riemanian products. 

\begin{Ex}
\label{embT2}
Consider $S^1$ as a submanifold of $\mathbbm{T}^2$ with the major radius $c$, and the radius of the cross-section $1$. The map $f:S^1 \rightarrow \mathbbm{T}^2$ given by $f(\theta):=(\theta,0)$ is an embedding, with a retraction $F(\theta,\phi):=\theta$. The metric tensors on $\mathbbm{T}^2$ and $S^1$ are
\begin{equation*}
g_{\mathbbm{T}^2}=~\left(\begin{array}{cc}
     1 & 0  \\
     0 & (c-\sin{\theta})^2
\end{array}\right)\quad {\textit and} \qquad
g_{S^1}=(1), 
\end{equation*}
respectively. The circle is clearly an isometrically embedded submanifold of $\mathbbm{T}^2$. Let $h(\theta,\phi)\in C^\infty(\mathbbm{T}^2)$ and $\xi\in \mathcal{H}_{\mathbbm{T}^2}=L^2({\Omega}_{\mathbbm{C}}(\mathbbm{T}^2))$, then compute 
\begin{align*}
   d_{\mathbbm{T}^2}F^*(f^*(h))=&~d_{\mathbbm{T}^2} h(\theta,0) 
   =
   \frac{\partial h}{\partial \theta}d\theta.
\end{align*}
The norm of commutators with the Hodge de Rham Dirac operators is given by: 
\begin{align*}
\|\ [\mathcal{D}_{\mathbbm{T^2}},F^*(h\circ f)]\ \|_{\mathcal{B}(H_{\mathbbm{T}^2})}=&~g^{\theta\theta}\left(\frac{\partial h}{\partial\theta}\right)^2 
=\|\ [\mathcal{D}_{S^1},{h}\circ f]\ \|_{\mathcal{B}(\mathcal{H}_{S^1})}. 
\end{align*}

\medskip 

Notice again that, as described in remark~\ref{rem: Rsub}, we have a bounded Fubini partial integration expectation $u_{\hat{r}}:L^2(\Omega_{\mathbbm{C}}(\mathbbm{T}^2))\to L^2(\Omega_{\mathbbm{C}}(S^1))$,  but $u_{\hat{r}}$ does not coincide with the usual pull-back $f^*$ on continuous forms.  

\medskip 

If we choose $u:=f^*$ the pull-back on forms, then $u$ is a densely defined unbounded operator from continuous forms on $T^2$ to $L^2(\Omega_{\mathbbm{C}}(S^1))$. To see this, let $h\in C(S^1)$ 
and consider the family of continuous functions on $S^1$ given by: 
\begin{equation*}
  \displaystyle  
  g_n(\varphi):=\left\{\begin{array}{cc}
       \sqrt{n\cos\left(\frac{n\pi}{2\varphi_0}\varphi\right)}  &, 
       \quad \varphi< \varphi_0/n  
       \\
        0 &, 
        \quad \varphi\geq \varphi_0/n
    \end{array}\right.
\end{equation*} 
where $\varphi_0$ is a small angle. One can check that
\begin{align*}
    \|hg_n\|_{L^2(T^2)}=&~\left(\int nh^2(\theta)\cos\left(\frac{n\pi}{2\varphi_0}\varphi\right)d\theta d\varphi\right)^{1/2} \\
    =&~2\sqrt{\frac{\varphi_0}{\pi}}\|h\|_{L^2(S^1)} \\
    <&~\infty.
\end{align*}
However, the pull-back of the function is unbounded i.e.
\begin{align*}
    \|f^*(hg_n)\|_{L^2(S^1)}=
    \int nh(\theta)^2 d\theta 
    =
    n\|h\|_{L^2(S^1)},
\end{align*}
which diverges as $n\rightarrow \infty$.
\end{Ex}

Finally, let us investigate the relation between the Riemannian isometric map $f:N\to M$ and the geodesic distance on the submanifold. 
The metric in the canonical spectral triple, which coincides with the geodesic distance computed from the metric tensor on the manifold $M$, is given by Connes' distance formula
\begin{align}
d_{\mathcal{D}_M}(y,y')=\sup \left\{|h (y)-h(y')| \ \big{|} \ h \in C^\infty(M), \ \|[\mathcal{D}_M,h] \|_{\mathcal{B}(\mathcal{H}_M)} \leq 1\right\}, 
\end{align}
for any $y,y' \in M$. 
Notice first the trivial fact that condition~ii) in theorem~\ref{Lhimma} implies the equality of Connes' distances (and hence geodesic distances) 
\begin{equation}
d_{\mathcal{D}_M}(f(x),f(x'))=d_{\mathcal{D}_N}(x,x'), \quad \forall x,x'\in N, 
\label{isometric}
\end{equation}
as it is expected in the special case of open Riemannian embeddings, that are already known to be global isometries.

\medskip 

From theorem~\ref{th: Tr} we obtain instead this ``variant'' of Connes' isometry
\begin{gather*}
\hat{d}_{\mathcal{D}_M}(f(x),f(x'))=d_{\mathcal{D}_N}(x,x'), \quad \forall x,x'\in N, \quad \text{where}
\\
\hat{d}_{\mathcal{D}_M}(y,y'):=\sup \{|h(y)-h(y')| \ | \  h\in F^*(C^\infty(N)), \  \|[\mathcal{D}_M, h]\|_{\mathcal{B}(\mathcal{H}_M)}\leq 1 \}
\end{gather*}
is a ``modified Connes' distance'' over the subalgebra $F^*(C^\infty(N))\subset C^\infty(M)$; 
this is in line with our definition of isometric sub-triples in section~\ref{Mor} (although the ``cutting projection'' that is used is not of the form $P:=u^*u$ with $u=f^*$ the usual pull-back morphism). 

\medskip 

Although we know the formula of distance on certain spectral triples, computing it is not very practical. The merit of employing such formula is that the norm equality in the theorem \ref{Lhimma} and \ref{th: Tr} can be extended to the notion of isometry on some noncommutative spectral triples.

\section{Spectral Subtriples}
\label{Mor}

An exact sequence of C$^*$-algebras induces morphisms between the state spaces of the algebras: 

\begin{center}
\begin{tikzcd}
 0\arrow{r}&\mathcal{I}\arrow[d, my mapsto]\arrow{r}{\iota} &\mathcal{A} \arrow{r}{\varphi} \arrow[d, my mapsto] & \mathcal{B} \arrow[d, my mapsto]\arrow{r} & 0 \\
& \mathscr{S}(\mathcal{I}) & \arrow{l}{\iota^*} \mathscr{S}( \mathcal{A})\arrow{l}  & \arrow{l}{\varphi^*} \mathscr{S}(\mathcal{B})
\end{tikzcd}
\end{center}
The induced map $\varphi^*$ is defined by
\begin{align*}
(\varphi^*\rho)(a)=&~\rho(\varphi(a)), \quad a\in \mathcal{A}, \ \rho\in \mathscr{S}(\mathcal{B}).
\end{align*}

Since $\varphi$ is surjective, by the GNS construction, one can show that the representation constructed form $\rho$ and $\varphi^*\rho$ are unitary equivalent \cite{Emch}. This makes $\pi_{\varphi^*\rho}$ irreducible, whenever $\rho$ is a pure state.
The irreducibility implies that the induced map $\varphi^*$ preserves the purity of states. 

\medskip 

Whenever the C*-algebras $\mathcal{A}, \mathcal{B}$ belong to the spectral triples $(\mathcal{A},\mathcal{H}_\mathcal{A}, \mathcal{D}_\mathcal{A})$ and $(\mathcal{B},\mathcal{H}_\mathcal{B}, \mathcal{D}_\mathcal{B})$, 
using Connes' metric \eqref{eq: Connes-d} one may naively define the isometric condition as a noncommutative version of \eqref{isometric}: 
\begin{equation*}
d_{\mathcal{D}_\mathcal{A}}(\varphi^*\rho,\varphi^*\rho')=d_{\mathcal{D}_\mathcal{B}}(\rho,\rho'), \quad \forall \rho,\rho'\in \mathscr{P}(\mathcal{B}).
\end{equation*}
However, since the Connes' distance gives the geodesic distance between any pair of pure states, the computation of Connes' distance, using the Dirac operator of the ambient space, without any constraint, would yield for example the distance of the segment in Figure \ref{Riemdist} instead of the distance along the arc.

\begin{figure}
    \centering
    \begin{tikzpicture}[scale=0.65]
  \draw (0,0) circle [radius=3 cm];
  \draw (0,3) node[above=0.1 cm]{$x$};
  \draw (3,0) node[right=0.1 cm]{$y$};
  \draw[teal] (0,3)node[shape=circle,fill=black, scale=0.3]{p} -- (3,0) node[shape=circle,fill=black, scale=0.3]{p}; 
  \draw [red] (3,0) arc[radius=3 cm,start angle=0,end angle=90];
\end{tikzpicture}
    \caption{Connes' distance on $S^1$ embedded in $\mathbbm{R}^2$ versus geodesic distance}
    \label{Riemdist}
\end{figure}
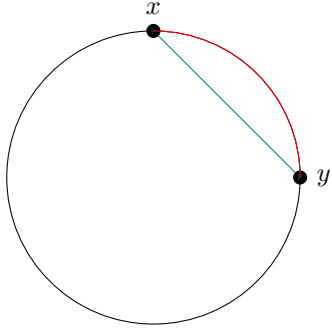

\medskip 

To attempt a first definition of isometric subtriple, we use the notions of morphism between spectral triples in \cite{Bertozzini_2012, Flo-Gho} with some modification.

\begin{Def} \label{def: iso}
\label{subtriple}
Suppose that $(\mathcal{A}_1,\mathcal{H}_1,\mathcal{D}_1)$ and $(\mathcal{A}_2,\mathcal{H}_2,\mathcal{D}_2)$ are spectral triples. 
A {\bf smooth naive morphism} $(\mathcal{A}_1,\mathcal{H}_1,\mathcal{D}_1)\xleftarrow{(\varphi,u)}(\mathcal{A}_2,\mathcal{H}_2,\mathcal{D}_2)$ consists of 
\begin{itemize}
\item[1.] 
a unital $*$-homomorphism $\varphi: \mathcal{A}_2 \rightarrow \mathcal{A}_1$ 
\item[2.] 
a (possibly unbounded) \footnote{
In the unbounded case, we require $\Dom(u)$ dense in $\mathcal{H}_2$, invariant for all $T\in\Omega_{\mathcal{D}_2}(\mathcal{A}_2)$ and 
the existence of a common invariant core for $u$ and $\mathcal{D}_2$, inside $\Dom(u)\cap\Dom({\mathcal{D}_2})$. 
} 
linear map $u:\mathcal{H}_2 \rightarrow \mathcal{H}_1$ 
\end{itemize}
such that 
\begin{itemize}
\item[3.] 
$u \pi_2(x)=\pi_1(\varphi(x)) u$, for all $x\in\mathcal{A}_2$, 
\item[4.] 
$\varphi(\mathcal{A}_2^\infty)\subset \mathcal{A}_1^\infty, \quad u(\mathcal{H}_2^\infty)\subset \mathcal{H}_1^\infty$.  
\end{itemize}
We say that a smooth morphism is an {\bf embedding of spectral triples} if: 
\begin{itemize}

\item[5.] 
$\varphi$ is surjective and 
$\mathcal{A}_1^\infty$ is a submanifold algebra of $\mathcal{A}_2^\infty$
\end{itemize}
Given an embedding of spectral triples 
$(\mathcal{A}_1,\mathcal{H}_1,\mathcal{D}_1)\xleftarrow{(\varphi,u)}(\mathcal{A}_2,\mathcal{H}_2,\mathcal{D}_2)$, 
we say that $(\mathcal{A}_1,\mathcal{H}_1,\mathcal{D}_1)$ is {\bf spectral subtriple} of $(\mathcal{A}_2,\mathcal{H}_2,\mathcal{D}_2)$ if one of the following further conditions are satisfied:  
\begin{itemize}
\item[a.] 
$(\varphi,u)$ is a {\bf Connes' isometry} (see e.g.~\cite[definition 2.3]{bassi} and also \cite{BCL11}): 
\begin{equation*}
d_{\mathcal{D}_2}(\varphi^*\rho,\varphi^*\rho')=d_{\mathcal{D}_1}(\rho,\rho'), \quad \forall \rho,\rho'\in\mathscr{P}(\mathcal{A}_1),
\end{equation*}
in which case $(\mathcal{A}_1,\mathcal{H}_1,\mathcal{D}_1)$ is called a {\bf Connes' isometric subtriple}; 
\item[b.] 
$(\varphi,u)$ is a {\bf Riemannian isometry} (see \cite{MarioPaschke_2004} and also \cite{Bertozzini_2012}):  
\begin{equation*}
u[\mathcal{D}_2,\pi_2(x)] =[\mathcal{D}_1,\pi_1(\varphi(x))]u, \quad \forall x\in \mathcal{A}_2,  
\end{equation*}
in which case $(\mathcal{A}_1,\mathcal{H}_1,\mathcal{D}_1)$ is called a {\bf Riemannian subtriple}; 
\item[c.]
$(\varphi,u)$ is a {\bf totally geodesic isometry} (see \cite{Paolo2005}): 
\begin{equation*}
u\mathcal{D}_2=\mathcal{D}_1u,
\end{equation*}
in which case $(\mathcal{A}_1,\mathcal{H}_1,\mathcal{D}_1)$ is a {\bf totally geodesic subtriple}, 

\item[d.]
$u$ is a partial isometry onto $\mathcal{H}_1$, with projection $P=u^*u\in \mathcal{B}(\mathcal{H}_2)$, and
\begin{equation*}
d_{P\mathcal{D}_2P}(\varphi^*\rho,\varphi^*\rho')= d_{\mathcal{D}_1}(\rho,\rho'), \quad \forall \rho,\rho'\in\mathscr{P}(\mathcal{A}_1), 
\end{equation*}
in wihich case $(\mathcal{A}_1,\mathcal{H}_1,\mathcal{D}_1)$ will be called an {\bf isometric subtriple}.
\end{itemize}
\end{Def}

Some important comments are necessary in order to explain the (temporary) choice and the limitations involved in the previous definition. Although some of the following arguments are very concise and not fully developed, we hope that they might provide a sufficient justification for the present choice of naive morphisms of spectral triples (with a possibly unbounded map $u$ in order to cover non trivial embeddings of Riemannian manifolds). 
\begin{itemize}

\item 
The introduction of the smoothness and submanifold conditions in points 4.~and 5.~above, serves to assure the ``smoothness and regularity'' of the morphisms: in the commutative case, we are not aware of any result generalizing Myers-Steenrod theorem and hence a Connes' metric isometry in point a.~above, might not necessarily imply smoothness in the case of isometric embeddings where the dimension of $N$ is strictly less than the dimension of $M$. 
\item 
As noted in \cite{F}, any embedding $N\xrightarrow{f} M$ of Riemannian manifolds, determines a canonical isomorphism of pre-Hilbert C*-modules over $C(N)$ 
\begin{equation*}
\Omega(f(N))\otimes_{C(N)}\Omega(f(N)^\perp)\simeq C(N)\otimes_{C(M)}\Omega(M). 
\end{equation*}
that is obtained by second quantisation of the orthogonal decomposition in \eqref{Decom}, where $\Omega(f(N))$ denotes the $C(N)$-module of continuous sections of the complexified of the Grassmann bundle of the manifold $f(N)$; $\Omega(f(N)^\perp)$ is the $C(N)$-module of continuous sections of the complexified of the Grassmann bundle of the normal bundle $\mathcal{N}_f$ and we indicate with $C(N)\otimes_{C(M)}\Omega(M)$ the \hbox{$C(N)$-module} of continuous sections of the restriction to $f(N)$ of the complexified Grassmann bundle of $M$. Upon completion under the natural $C(N)$-valued inner products, the above isomorphism extends to Hilbert C*-modules over the C*-algebra $C(N)$. 
There is natural bounded ``restriction'' morphism $\Omega(M)\to C(N)\otimes_{C(M)}\Omega(M)$ via $\sigma\mapsto 1_{C(N)}\otimes \sigma$;\footnote{
This is a morphism of modules over different rings $C(M)$ and $C(N)$ respectively with $\varphi$ twisting the action. 
} 
and there is a natural bounded ``projection'' morphism $\Omega(f(N))\otimes_{C(N)}\Omega(f(N)^\perp)\to \Omega(f(N))\simeq \Omega(N)$ on the ortocomplemented first factor $\Omega(f(N))\otimes_{C(N)}1_{\Omega(f(N)^\perp)}$.  
Making use of the Hodge de Rham spectral module triples introduced in \cite{F}, we get a well defined pull-back morphism $(\varphi,u)$ with bounded map $\Omega(M)\xrightarrow{u}\Omega(N)$. 

\item 

Unfortunately the previous nice bounded description at the level of Hilbert C$^*$-modules does not immediatley pass to the Hilbert spaces $L^2(\Omega_{\mathbbm{C}}(M))$ and $L^2(\Omega_{\mathbbm{C}}(N))$ induced by Riemann integration with respect to the respective Riemannian volume measures: the pull-back of forms is not well-defined as a bounded map between $L^2(\Omega_{\mathbbm{C}}(M))\to L^2(\Omega_{\mathbbm{C}}(N))$, as can be seen from the explicit example~\ref{embT2}.

It is anyway possible to obtain an unbounded pull-back map at the level of Hilbert spaces $L^2(\Omega_{\mathbbm{C}}(M))\xrightarrow{u} L^2(\Omega_{\mathbbm{C}}(N))$: using the compactness of the submanifold $f(N)\subset M$, we can find a tubular normal neighborhood $T_r$ of $f(N)$ with sufficiently small but finite radius $r>0$. 
Using the normal geodesic flow induced by the exponential map of the tubular neighborhood, we obtain a description of the tubular neighborhood $T_r$ as a foliation over $f(N)$ with fibres of dimension $\dim(M)-\dim(N)$. 

Utilising the co-area formula, generalizing Fubini-Tonelli theorem to foliations~\cite{Nicolaescu-notes}, we obtain a well-defined (fibre-normalised) partial integration bounded map $L^2(\Omega_{\mathbbm{C}}(T_r))\to L^2(\Omega_{\mathbbm{C}}(N))$ that composed with the bounded restriction map $L^2(\Omega_{\mathbbm{C}}(M))\to L^2(\Omega_{\mathbbm{C}}(T_r))$ provides a bounded expectation map $u_r: L^2(\Omega_{\mathbbm{C}}(M))\to L^2(\Omega_{\mathbbm{C}}(N))$. 
One can rigorously define our unbounded ``restriction'' operator $u: L^2(\Omega_{\mathbbm{C}}(M))\to L^2(\Omega_{\mathbbm{C}}(N))$ as $\lim_{r\to 0} u_r$, noting that on the dense subspace of ``continuous'' sections it is well-defined and, on continuous sections reproduces the usual restriction described before at the C*-module level. 

\item 

In the context here described, where $\varphi$ is thought to be a noncommutative generalization of the usual commutative ``restriction'' homomorphism between continuous functions and $u$ is supposed to generalize the usual ``pull-back'' of continuous differential forms (that provide only a dense subspace inside the Hilbert space of square integrable sections of the Grassmann bundles), the imposition of the boundedness of $u$ and its partial-isometric nature, is an extremely restrictive requirement that essentially limits the discussion (in commutative Riemannian situation) to the embedding of some sub-family of  connected components into a disconnected manifold.  
Hence the definition in point d.~above, has a very limited scope neverless, in the almost-commutative setting, a few interesting cases are described in the examples. 
\item 
As already mentioned in some detail in remark~\ref{rem: Rsub}, the partial isometry property of $u$ will be available also between Hilbert spaces of square integrable sections of Grassmann bundles, whenever a different type of morphism, based on conditional expectations (instead of homomorphism) between algebras, is used; the typical case would be the canonical expectations 
\begin{equation*}
L^2(\Omega_{\mathbbm{C}}(N_1))\xleftarrow{u_1} L^2(\Omega_{\mathbbm{C}}(N_1\times N_2)) \xrightarrow{u_2}L^2(\Omega_{\mathbbm{C}}(N_2))
\end{equation*}
available, via normalised partial integration, in the product of compact manifolds $N_1\times N_2$ and similarly also in our previously defined bounded operator $u_r: L^2(\Omega_{\mathbbm{C}}(T_r))\to L^2(\Omega_{\mathbbm{C}}(N))$, for $r>0$. 

We are not immediately entering into this topic here, leaving further detailed discussion of these points to subsequent work; we just mention that a more satisfactory description of morphisms (at least for commutative Riemannian spectral triples) will be achieved using, as already hinted in~\cite{F}, \textit{co-spans of naive morphism} of spectral triples: 
\begin{equation*}
(\mathcal{A}_N,\mathcal{H}_N,\mathcal{D}_N) \xrightarrow{} 
(C_b(T_r),L^2(\Omega_{\mathbbm{C}}(T_r)),\mathcal{D}_M|_{T_r}) \xleftarrow{}
(\mathcal{A}_M,\mathcal{H}_M,\mathcal{D}_M), 
\end{equation*}
eventually passing to the ``inductive limit'' of triples, as in~\cite{Flo-Gho}, for $r\to 0$. 
\end{itemize}

\bigskip

We collect here some immediate consequences and remarks related to the previous definitions. 
\begin{Prop}
Totally geodesic isometries are Riemannian isometries.

\medskip 

Whenever $uu^*=\text{Id}_{\mathcal{H}_1}$ (and hence $u$ is a partial isometry restricting to a unitary between $\mathcal{H}_1$ and $u^*u(\mathcal{H}_1)\subset\mathcal{H}_2$) a Riemannian isometry specially implements a unital $*$-homomorphism extending $\varphi$ to the complex Clifford algebras: 
\begin{equation*}
\CCl_{\mathcal{D}_1}(\mathcal{A}_1) \xleftarrow{u - u^* }\CCl_{\mathcal{D}_2}(\mathcal{A}_2), 
\quad \quad 
uTu^*\mapsfrom T.
\end{equation*}
In the reverse direction, we have an inclusion unital $*$-homomorphism: 
\begin{equation*}
\CCl_{\mathcal{D}_1}(\mathcal{A}_1) \xrightarrow{u^*- u} \CCl_{\mathcal{D}_2}(\mathcal{A}_2), 
\quad \quad 
S \mapsto u^*Su,  
\end{equation*}
isomorphically mapping $\CCl_{\mathcal{D}_1}(\mathcal{A}_1)$ onto the subalgebra 
\begin{equation*}
P(\CCl_{\mathcal{D}_2}(\mathcal{A}_2))P\subset \CCl_{\mathcal{D}_2}(\mathcal{A}_2), 
\end{equation*} 
where $P:=u^*u$ is the Hilbert projection onto $u^*u(\mathcal{H}_2)$. 
The map $T\mapsto PTP$ is a conditional expectation of $\CCl_{\mathcal{D}_2}(\mathcal{A}_2)$ onto $u\CCl_{\mathcal{D}_1}(\mathcal{A}_1)u^*=P\CCl_{\mathcal{D}_2}(\mathcal{A}_2)P$. 
\end{Prop}
\begin{proof}
It is immediate that $u\mathcal{D}_2=\mathcal{D}_1u$ and $u\pi_2(x)=\pi_1(\varphi(x))u$, for $x\in\mathcal{A}_2$, imply 
\begin{equation*}
u[\mathcal{D}_2,\pi_2(x)]=[\mathcal{D}_1,\pi_1(\varphi(x))]{u}. 
\end{equation*}
The remaining statements, when $u$ is a partial isometry, are verified using $uu^*=\text{Id}_{\mathcal{H}_1}$, $P^2=P:=u^*u$ together with $u[\mathcal{D}_2,\pi_2(x)]u^*=[\mathcal{D}_1,\pi_1(\varphi(x))]$, for all $x\in\mathcal{A}_2$. 
\end{proof}

\begin{Rem}
Always under the condition $uu^*=\text{Id}_{\mathcal{H}_1}$, a totally geodesic isometry immediately implies $u\mathcal{D}_2u^*=\mathcal{D}_1$ and hence, ``compressing'' the spectral triple $(\mathcal{A}_2,\mathcal{H}_2,\mathcal{D}_2)$ by the orthogonal projection $P:=u^*u$, we get a spectral triple $(P \pi_2(\mathcal{A}_2) P, P(\mathcal{H}_2), P\mathcal{D}_2 P)$ that is isomorphic (unitarily equivalent) to the spectral triple $(\mathcal{A}_1,\mathcal{H}_1,\mathcal{D}_1)$ and $u^*:\mathcal{H}_1\to P(\mathcal{H}_2)$ is a unitary intertwiner between such isomorphic spectral triples. 
\end{Rem}

\begin{Rem}
It is not possible to immediately obtain implications between Connes' metric isometries and (totally geodesic) Riemannian isometries since the notion of Connes' metric isometry relies only on the existence of a unital $*$-homomorphism $\varphi$ at the level of algebras and (even assuming the existence of a lifting of such homomorphism to the complex Clifford algebras) this does not necessarily imply the existence of a spatial implementation of $\varphi$ (or its Clifford extension) between the Hilbert spaces of the spectral triples. 

\medskip 

Things will be different if we assume the existence of a Clifford algebra extension of the homomorphism $\varphi$ and we further introduce the requirement of existence of cyclic separating vector for von Neumann Clifford algebras as in \cite{LORD2012}. We will pursue this direction elsewhere. 
\end{Rem}
 
One can see that totally geodesic isometries are the strongest among these introduced notions: from the proposition above, a totally geodesic subtriple is also a Riemannian subtriple; furthermore, as it will be proved in the next section, a totally geodesic subtriple is (under suitable conditions) also an isometric sub-triple. The implications are depicted in the following diagram:  
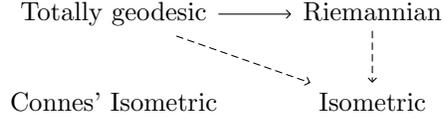
\begin{equation*}
\begin{tikzcd}
\text{Totally geodesic} \arrow[dr, dashed] \arrow[r]  &  \text{Riemannian} \arrow[d, dashed]
\\
\text{Connes' Isometric} & \text{Isometric} 
\end{tikzcd}
\end{equation*}
The meaning of the dashed arrows will be clear in section \ref{Iso}.
For the moment we do not provide implications from metric isometries to (totally geodesic) Riemannian isometries, although (under certain conditions) we expect to have some results in this direction (see the last remark above). 

\begin{Ex}
{\bf Finite spectral triples}\\
For $N\geq 2$, an $N-$point space is given by a spectral triple 
\begin{equation*}
F_N=(\bigoplus^N_{i=1}\mathbbm{C},\mathbbm{C}^{N(N-1)}, \mathcal{D}_N) 
\end{equation*}
with the representation \cite{AlmostNCG}
\begin{align*}
\pi_{N}(a_1, \ldots, a_N)=&~\pi_{N-1}(a_1, \ldots, a_{N-1})\\
&\oplus\left(\begin{array}{cc} a_1 & 0 \\ 0 & a_N \end{array}\right)\oplus \cdots \oplus\left(\begin{array}{cc} a_{N-1} & 0 \\ 0 & a_N \end{array}\right)
\end{align*}
where $a=(a_1, \ldots, a_{N})\in \bigoplus^N_{i=1}\mathbbm{C}$, and the Dirac operator 
\begin{equation*}
\mathcal{D}_{N}=\mathcal{D}_{N-1}\oplus\bigoplus^{N-1}_{j=1}\mathcal{D}_{j,N}, 
\end{equation*} 
where
\begin{equation*}
    \mathcal{D}_{j,N}=\left(\begin{array}{cc} 0 & x_{j,N} \\ \bar{x}_{j, N} & 0 \end{array}\right),
\end{equation*}
and $x_{j,N}$ is a complex number. Note that
\begin{equation*}
\pi_2(a_1, a_2)=\left(\begin{array}{cc} a_1 & 0 \\ 0 & a_2 \end{array}\right)~~{\rm and}~~
\mathcal{D}_2=\left(\begin{array}{cc} 0 & x_{1,2} \\ \bar{x}_{1,2} & 0 \end{array}\right).
\end{equation*}
We define a map $\varphi$ to be 
\begin{equation*}
    \varphi(a_1,a_2, ..., a_N):=(a_1,a_2,..., a_{N-1}),
\end{equation*}
which is surjective. Derivation on $\mathbbm{C}^N$ (diagonal $n\times n$ complex matrices) is a linear map satisfying Leibniz rule, i.e. for any $B\in\mathbbm{C}^N$ a derivation $\delta_A$ is given by
\begin{equation*}
\delta_{A}B=[A,B],
\end{equation*}
which is well-defined only when $A\in \mathbbm{C}^N$. The set ${\rm Der}(\mathbbm{C}^N)$ contains only zero maps, and hence ${\rm Der}_\varphi(\mathbbm{C}^N)={\rm Der}(\mathbbm{C}^N)$. The induced trivial zero map $\varphi_*:{\rm Der}_\varphi(\mathbbm{C}^N) \rightarrow {\rm Der}(\mathbbm{C}^{N-1})$ is given by $\varphi_*\delta_AB=[\varphi(A),B]~,$
for $A\in \mathbf{D}_N(\mathbbm{C})$ and $B \in \mathbbm{C}^{N-1}$. Since $\varphi$ is surjective, the induced map taking derivations on $\mathbbm{C}^N$ to derivations on $\mathbbm{C}^{N-1}$ is also surjective. Therefore, $\mathbbm{C}^{N-1}$ is a submanifold algebra of $\mathbbm{C}^N$. 
Let us choose $u:\mathbbm{C}^{N(N-1)}\rightarrow \mathbbm{C}^{(N-1)(N-2)}$ to be a projection mapping $\psi=\xi\oplus\eta \mapsto \xi$, for $\xi \in \mathbbm{C}^{(N-1)(N-2)}$ and $\eta\in \mathbbm{C}^{2(N-1)}$, then for $a\in \mathbbm{C}^N$
\begin{align*}
u\pi_N(a)\psi=
\pi_{N-1}(a_1, \dots, a_N-1)\xi 
=
\pi_{N-1}(\varphi(a))u\psi,
\end{align*}
and one can show that $F_2$ is a Riemannian subtriple of $F_3$ i.e. 
\begin{align*}
    u[\mathcal{D}_{3},\pi_{3}(a_1,a_2, 
    a_3)] 
&= u\left[
    \begin{pmatrix}
    \,0 & x_{1,2}(a_2-a_1)\, \\ \,-\bar{x}_{1,2}(a_2-a_1) & 0\,~ 
    \end{pmatrix} \right. 
   \\ 
   & \quad \quad \left. \oplus
    \begin{pmatrix}
    \,0 & x_{1,3}(a_3-a_1)\, \\ \,-\bar{x}_{1,3}(a_3-a_1) & 0\,~ 
    \end{pmatrix}\right.  
    \\ 
    & \left.
  \quad  \quad \oplus
    \begin{pmatrix}
    \,0 & x_{2,3}(a_3-a_2)\, \\ \,-\bar{x}_{2,3}(a_3-a_2) & 0\,~ 
    \end{pmatrix}\right]
    \\
    &=
    \begin{pmatrix}
    \,0 & x_{1,2}(a_2-a_1)\, \\ \,-\bar{x}_{1,2}(a_2-a_1) & 0\,~ 
    \end{pmatrix}u 
    \\
    &=[D_2,\pi_2(\varphi(a_1,a_2,a_3))]u. 
\end{align*}
By induction, one can show that for $N\in\mathbbm{N}$ and $a\in \mathbbm{C}^N$ 
\begin{equation*}
u[D_N,\pi_N(a)]=[D_{N-1},\pi_{N-1}(\varphi(a))]u,
\end{equation*}
hence $F_{N-1}$ is a subtriple of $F_N$. 
\end{Ex}
\begin{Ex}{\bf Finite-dimensional C$^*$-algebra}\\
Let $(\mathcal{A},\mathcal{H}_\mathcal{A},\mathcal{D}_\mathcal{A})$ and $(\mathcal{B},\mathcal{H}_\mathcal{B},\mathcal{D}_\mathcal{B})$ be finite spectral triples, and $\varphi:\mathcal{A}\rightarrow \mathcal{B}$ be a surjective homomorphism between $\mathcal{A}$ and $\mathcal{B}$. 

Then by \cite{SIGMA16} $\mathcal{B}$ is a submanifold algebra of $\mathcal{A}$. Suppose the representation space of $\mathcal{A}$ and $\mathcal{B}$ are $\mathbbm{C}^m$ and $\mathbbm{C}^n$, respectively. We define $u$ to be 
 \begin{equation*}
     u=\begin{pmatrix} \mathbbm{1}_n & 0
     \end{pmatrix}_{n\times m}
 \end{equation*}
It is easy to see that $u^*\pi_\mathcal{B}(\varphi(a))u$ is a subrepresentation of $\pi_\mathcal{A}$, therefore, the condition (2) is satisfied. If one chooses $\mathcal{D}_\mathcal{A}$ and $\mathcal{D}_\mathcal{B}$ satisfying one of the conditions a, b, c or d of Definition \ref{def: iso} then $(\mathcal{B},\mathcal{H}_\mathcal{B},\mathcal{D}_\mathcal{B})$ is a subtriple of $(\mathcal{A},\mathcal{H}_\mathcal{A},\mathcal{D}_\mathcal{A})$.
\end{Ex}
\begin{Ex} {\bf Almost commutative geometry}\\
\label{Almost}
Consider a discrete spectral triple $F_{\rm ED}:=(\mathbbm{C}^2,\mathbbm{C}^4, D_F,\gamma_F)$,
where
\begin{equation*}
    D_F=\begin{pmatrix} 0 & d & 0 & 0 \\
                 \bar{d} & 0 & 0 & 0 \\
                 0 & 0 & 0 & \bar{d} \\ 
                 0 & 0 & d & 0
        \end{pmatrix}, ~~\gamma_F=\begin{pmatrix} -1 & 0 & 0 & 0 \\
                 0 & 1 & 0 & 0 \\
                 0 & 0 & 1 & 0 \\ 
                 0 & 0 & 0 & -1
        \end{pmatrix}.
\end{equation*} 
The spectral triple is used in the formulation of electrodynamics. The product of this spectral triple with a canonical spectral triple is called an almost commutative spectral triple, which describes electromagnetic interaction on curved space
\begin{equation*}
M\times F_{\rm ED}:=\left(C^\infty(M,\mathbbm{C}^2),L^2(M,S\otimes\mathbbm{C}^4),\mathcal{D}_M\otimes\mathbbm{1}_4+\gamma_M\otimes \mathcal{D}_F\right),
\end{equation*}
where $\gamma_M$ is the grading operator on $M$. 

To show that the canonical spectral triple is a Riemannian subtriple of $M\times F_{\rm ED}$, suppose that $f \in C^\infty(M)$ and $\psi \in L^2(M,S)$, we define a map $\varphi: C^\infty(M,\mathbbm{C}^2) \rightarrow  C^\infty(M)$ by $\varphi(f_1,f_2):=f_1$, and $u$ is given by the tensor product of identity and $1\times 4$ matrix i.e.
 \begin{equation*}
    u={\rm Id}\otimes v,
\end{equation*}
where $v=\left(\begin{array}{cccc} 1 & 0 & 0 & 0 \end{array}\right).$
Then we have
\begin{align*}
u\pi'(f_1,f_2)=
f_1{\rm Id}\otimes v 
 =
 \pi(\varphi(f_1,f_2))u,    
\end{align*}
for the representations $\pi$ and $\pi'$ of $L^2(M,S\otimes\mathbbm{C}^2)$ and $L^2(M,S\otimes\mathbbm{C}^4)$, respectively. One can show that
\begin{align*}
u[\mathcal{D}_M\otimes\mathbbm{1}_4+&\gamma_M\otimes\mathcal{D}_F,\pi'(f_1,f_2)]\\
=&~u\left[\mathcal{D}_M\otimes\mathbbm{1}_4,{\rm Id}\otimes\begin{pmatrix} f_1\mathbbm{1}_2 & 0 \\ 0 & f_2\mathbbm{1}_2 \end{pmatrix}\right] \\
&~+ u\left[\gamma_M\otimes\mathcal{D}_F,{\rm Id}\otimes\begin{pmatrix} f_1\mathbbm{1}_2 & 0 \\ 0 & f_2\mathbbm{1}_2 \end{pmatrix}\right] \\
 =&~u\left[\mathcal{D}_M\otimes\mathbbm{1}_4,{\rm Id}\otimes\begin{pmatrix} f_1\mathbbm{1}_2 & 0 \\ 0 & f_2\mathbbm{1}_2 \end{pmatrix}\right]\\
  =&~u(-i\gamma^k)\otimes\begin{pmatrix} \partial_kf_1\mathbbm{1}_2 & 0 \\ 0 & \partial_kf_2\mathbbm{1}_2 \end{pmatrix}\\
  =&~(\mathcal{D}_M f_1)\otimes v \\
  =&~[\mathcal{D}_M,\pi(\varphi(f_1,f_2))]u.
\end{align*}
\end{Ex}
Hence, the canonical triple is a subtriple of $M\times F_{ED}$. All the examples above show that the previous definitions can work at least in the almost commutative setting. Although it is possible that such definitions might be compatible with a more general class of noncommutative spectral triples, examples for such cases are needed and a fully noncommutative geometric treatment of spectral subtriples might require the usage of suitable bimodules in place of 
surjective homomorphism. 

\section{Dirac Decompositions and Isometric Subtriples} \label{Iso}
From here on, we will investigate embeddings of spectral triples in the extremely restrictive cases where $u:\mathcal{H}_2\rightarrow \mathcal{H}_1$ is a partial isometry.

\begin{Lem}
\label{Pcomm} Let $(\mathcal{A}_1,\mathcal{H}_1,\mathcal{D}_1)\xleftarrow{(\varphi,u)} (\mathcal{A}_2,\mathcal{H}_2,\mathcal{D}_2)$ be an embedding of spectral triples as in definition \ref{subtriple}. 
If $u:\mathcal{H}_2\to\mathcal{H}_1$ is partial isometry with projection $P:=u^*u\in\mathcal{B}(\mathcal{H}_2)$, then $[\pi_2(x),P]=0$, for all $x\in \mathcal{A}_2$.
\end{Lem}
\begin{proof}
From the definition of smooth morphism in definition~\ref{subtriple} one has:  
\begin{equation}
u^*\pi_1(\varphi(x))u=u^*u\pi_2(x)=P\pi_2(x), \quad\quad \forall x\in\mathcal{A}_2. \label{YibYib}
\end{equation}
Conjugating \eqref{YibYib} one obtains also $\pi_2(x)P=u^*\pi_1(\varphi(x))u$. 
\end{proof}

Let $\rho,\rho'\in\mathscr{P}(\mathcal{A}_1)$ be pure states, consider the metrics
\begin{align}
d_{\mathcal{D}_1}(\rho,\rho')=&~\sup\{|\rho(a)-\rho'(a)|~\big{|}~  a\in \mathcal{A}_1, \|[\mathcal{D}_1,\pi_1(a)]\|\leq 1\}\nonumber \\     
     \leq&~\sup\{|\rho(a)-\rho'(a)|~\big{|}~  a\in \mathcal{A}_1, \|u^*[\mathcal{D}_1,\pi_1(a)]u\|\leq 1\} \nonumber\\
     =&~\sup\{|\rho(\varphi(b))-\rho'(\varphi(b))|~\big{|}~  b\in \mathcal{A}_2, \|u^*[\mathcal{D}_1,\pi_1(\varphi(b))]u\|\leq 1\}\nonumber\\
     =&~\sup\{|\rho(\varphi(b))-\rho'(\varphi(b))|~\big{|}~  b\in \mathcal{A}_2, \|u^*[\mathcal{D}_1,\pi_1(\varphi(b))]uu^*u\|\leq 1\}\nonumber\\
     =&~\sup\{|\rho(\varphi(b))-\rho'(\varphi(b))|~\big{|}~ b\in \mathcal{A}_2, \|P[\mathcal{D}_2,\pi_2(b)]P\|\leq 1\} \nonumber \\
     =&~\sup\{|\varphi^*\rho(b)-\varphi^*\rho'(b)|~\big{|}~ b\in \mathcal{A}_2, \|[P\mathcal{D}_2P,\pi_2(b)]\|\leq 1\}\nonumber \\
     =&~ d_{P\mathcal{D}_2P}(\varphi^*\rho,\varphi^*\rho'). \label{d2}
\end{align}
Note that we have used the definition of Riemannian spectral triple and 
$\|u^*[\mathcal{D}_1,\pi_1(a)]u\|\leq\|[\mathcal{D}_1,\pi_1(a)]\|$ or equivalently, by Lemma \ref{Pcomm}, $\|[P\mathcal{D}_2P,\pi_2(b)]\|\leq\|[\mathcal{D}_1,\pi_1(\varphi(b))]\|$. 

\medskip

\begin{Rem}
Under the condition that $u$ is a partial isometry, we have that the isometric condition $\|[P\mathcal{D}_2P,\pi_2(x)]\|=\|[\mathcal{D}_1,\pi_1(\varphi(x)]\|$, for Riemannian morphisms, is always satisfied: we know that isomorphisms of C*-algebras are norm-isometric; we know from the previous section that $T\mapsto u^*Tu$ provides an isomorphism between the Clifford algebra $\Omega_{\mathcal{D}_1}(\mathcal{A}_1)$ and $\Omega_{P\mathcal{D}_2P}(P\pi_2(\mathcal{A}_2)P)$; 
hence this isomorphism will extend to the C*-Clifford algebras (the norm completion of $\Omega_{\mathcal{D}}(\mathcal{A})$ in the operator norm) and 
this isomorphism will be necessarily operator norm isometric hence the metric condition is valid. 
\end{Rem}

\begin{Lem}\label{ccal} If $(\mathcal{A}_1,\mathcal{H}_1,\mathcal{D}_1)$ is a Riemannian subtriple of $(\mathcal{A}_2,\mathcal{H}_2,\mathcal{D}_2)$, with a projection $P=u^*u$. 
Let $\rho,\rho'\in\mathscr{P}(\mathcal{A}_1)$, if $\|[P\mathcal{D}_2P,\pi_2(b)]\|=\|[\mathcal{D}_1,\pi_1(\varphi(b))]\|$, for all $b\in \mathcal{A}_2$, then  
\begin{equation*}
d_{P\mathcal{D}_2P}(\varphi^*\rho,\varphi^*\rho')=d_{\mathcal{D}_1}(\rho,\rho').
\end{equation*}
\end{Lem}
\begin{proof}
This follows immediately from the previous calculation~\eqref{d2}, noting that the first and the sixth line there are equal as a consequence of the requirement $\|[P\mathcal{D}_2P,\pi_2(b)]\|=\|[\mathcal{D}_1,\pi_1(\varphi(b))]\|$, for all $b\in \mathcal{A}_2$ and the surjectivity of $\varphi$. 
\end{proof}

\medskip

If the projection $P$ defined above commutes with the Dirac operator i.e. $[\mathcal{D}_2,P]=0$, then
\begin{align*}
     \mathcal{D}_2&= \mathcal{D}_2P\oplus\mathcal{D}_2({\rm Id}-P) \\
                &=P\mathcal{D}_2P\oplus({\rm Id}-P)\mathcal{D}_2({\rm Id}-P) \\
                &=\tilde{\mathcal{D}}_1\oplus\mathcal{D}_{\rm norm},
\end{align*}
where we define $\tilde{\mathcal{D}}_1:=P\mathcal{D}_2P$ and $\mathcal{D}_{\rm norm}:=({\rm Id}-P)\mathcal{D}_2({\rm Id}-P)$. 
In the following theorem we generalize the condition employed in Theorem \ref{Lhimma}.

\begin{Thm}
\label{Th401}
Let $(\mathcal{A}_1,\mathcal{H}_1,\mathcal{D}_1)$ be a Riemannian subtriple of $(\mathcal{A}_2,\mathcal{H}_2,\mathcal{D}_2)$, with projection $P=u^*u$, and $[\mathcal{D}_2,P]=0$, then
\begin{description}
\item i) $\|[\mathcal{D}_2,\pi_2(b)]\|_{P\mathcal{H}_2}=\|[P\mathcal{D}_2P,\pi_2(b)]\|_{\mathcal{H}_2}$, for all $b\in\mathcal{A}_2$   
\item ii) If the Dirac satisfies $ \|[P\mathcal{D}_2P,\pi_2(b)]\|_{\mathcal{H}_2}=\|[\mathcal{D}_1,\pi_1(\varphi(b))]\|_{\mathcal{H}_1}$ for all $b \in \mathcal{A}_2$, then the subtriple is isometric.
\end{description}
\end{Thm}
\begin{proof} 
To prove i) suppose $\psi=\xi+\eta$ where $\xi\in P\mathcal{H}_2,$ and $\eta\in (\mathbbm{1}-P)\mathcal{H}_2$. Let $b \in \mathcal{A}_2$ and $[\mathcal{D}_2,b]\psi=\tilde{\xi}+\tilde{\eta}$, for some $\tilde{\xi}\in P\mathcal{H}_2$ and $\tilde{\eta}\in (\mathbbm{1}-P)\mathcal{H}_2$, then 
\begin{align*}
    [\mathcal{D}_2,b]\xi=
    [\mathcal{D}_2,b]P\xi 
                        =
                        P[\mathcal{D}_2,b]\xi 
                        =
                        \tilde{\xi},
\end{align*}
hence $\displaystyle \|[\mathcal{D}_2,b]\|_{P\mathcal{H}_2}=\sup_{\xi\not=0}\|\tilde{\xi}\|/\|\xi\|$. 
Then consider the norm 
\begin{align*}
\|[P\mathcal{D}_2P,b]\|_{\mathcal{H}_2}=\sup_{\psi\not=0}F(\xi,\eta),
\end{align*}
where $F(\xi,\eta)=\|\tilde{\xi}\|/\sqrt{\|\xi\|^2+\|\eta\|^2}$. Note that for the case $\xi=0$, $F(0,\eta)=0$ because $P[\mathcal{D}_2,b]P\eta=0$, and for the case $\xi\not=0$, one has $F(\xi,0)>F(\xi,\eta)$. 

The supremum can be obtained only when $\xi\not=0$ and $\eta=0$, therefore, $\|[\mathcal{D}_2,b]\|_{P\mathcal{H}_2}=\|[P\mathcal{D}_2P,b]\|_{\mathcal{H}_2}$.

\medskip 

To prove $ii)$, suppose the Dirac operator satisfies the assumption, by lemma \ref{ccal} the subtriple is isometric. 
\end{proof}  

\begin{Cor} 
A totally geodesic subtriple with $uu^*=\text{Id}_{\mathcal{H}_1}$ is an isometric subtriple.
\end{Cor}
\begin{proof}
Let $(\mathcal{A}_1,\mathcal{H}_1,\mathcal{D}_1)$ be a totally geodesic subtriple of $(\mathcal{A}_2,\mathcal{H}_2,\mathcal{D}_2)$, consider
\begin{equation*}
    P\mathcal{D}_2=~u^*u\mathcal{D}_2=u^*\mathcal{D}_1u=\mathcal{D}_2u^*u=\mathcal{D}_2P.
\end{equation*}
So, $[\mathcal{D}_2,P]=0$, and since $uu^*=\text{Id}_{\mathcal{H}_1}$, we have 
\begin{equation*}
\|[\mathcal{D}_1,\pi_1(a)]\|\leq\|uu^*[\mathcal{D}_1,\pi_1(a)]uu^*\|\leq \|u[\mathcal{D}_1,\pi_1(a)]u^*\|.
\end{equation*}
Recalling that $\|u[\mathcal{D}_1,\pi_1(a)]u^*\|\leq \|[\mathcal{D}_1,\pi_1(a)]\|$, we see that $(\mathcal{A}_1,\mathcal{H}_1,\mathcal{D}_1)$ 
is an isometric subtriple.
\end{proof}

\begin{Ex} 
{\bf Almost commutative space} \\
In the previous section we have shown that a canonical spectral triple is a Riemannian subtriple of an almost commutative geometry. Suppose we define $\varphi$ and $u$ as in Example \ref{Almost}, we have $u^*={\rm Id}\otimes v^*$. The projection $P=u^*u={\rm Id}\otimes P_{v}$, where
\begin{equation*}
    P_{v}=v^*v=\begin{pmatrix} \mathbbm{1}_2 & 0\\
                          0 & 0 
\end{pmatrix}
\end{equation*}
commutes with the Dirac operator. One obtains
\begin{align*}
    \tilde{\mathcal{D}}_1=&~\mathcal{D}_M\otimes P_{v} +
   \gamma\otimes \begin{pmatrix} 0 & d  & 0 & 0 \\
                                    \bar{d} & 0 & 0 & 0 \\
                                    0 & 0 & 0 & 0 \\
                                    0 & 0 & 0 & 0 \end{pmatrix}.
\end{align*}
Straight forward calculation shows that $\|[\tilde{\mathcal{D}}_1,\pi'(f_1,f_2)]\|=\|[\mathcal{D}_M,\pi(\varphi(f_1,f_2))]\|$, hence by Theorem \ref{Th401} the canonical triple is an isometric submanifold. 
One may prove this fact independently using the metric. Let $p \in M$, an evaluation map $\varepsilon_p: C^\infty(M) \rightarrow \mathbbm{C}$ is a pure state associated with the algebra of smooth functions. The induced map between pure states yields
\begin{align*}
    \varphi^*\varepsilon_p(f_1,f_2)=&~\varepsilon_p(\varphi(f_1,f_2)) \\
                               =&~\varepsilon_p(f_1)\\
                               =&~f_1(p) \\
                               =&~\varepsilon_p\otimes\omega(f_1,f_2) \\
                               =&~
                               \frac{1}{2}\,{\rm tr}\left[P_v \begin{pmatrix}f_1(p)\mathbbm{1}_2 & 0 \\
                    0 & f_2(p)\mathbbm{1}_2 
                   \end{pmatrix}\right],
\end{align*}
where $\omega$ is the pure state associated with $P_v$. Note that $\omega$ is independent with choice of $p$, by Theorem 2 from \cite{Martinetti:2001fq}
\begin{equation}
d(\varphi^*\varepsilon_p,\varphi^*\varepsilon_q)=d(\varepsilon_p\otimes\omega,\varepsilon_q\otimes\omega) =d(\varepsilon_p,\varepsilon_q). 
\end{equation}
The canonical triple is an isometric subtriple of the almost commutative spectral triple.
\end{Ex}

\section*{Acknowledgement}
This research project is supported by CMU Junior Research Fellowship Program, project code: JRCMU2564\_046. P. Bertozzini thanks his longtime collaborator R.Conti for several online relevant discussions.


{\small

\begin{thebibliography}{}

\bibitem
{Aschieri_2005}\hfill

Aschieri P, Blohmann C, Dimitrijevic M, Meyer F, Schupp P, Wess J (2005)
\emph{A Gravity Theory on Noncommutative Spaces} 
\textit{Class Quant Grav} 22(17):3511
\href{https://doi.org/10.1088/0264-9381/22/17/011}{https://doi.org/10.1088/0264-9381/22/17/011}
\ 
\href{https://arxiv.org/abs/hep-th/0510059}{arXiv:hep-th/0510059}

\bibitem
{Aschieri-book}\hfill

Aschieri P (2009)
\emph{Noncommutative Spacetimes} 
\textit{Lecture Notes in Physics} 774 
Springer

\bibitem
{bassi}\hfill

Bassi J, Conti R (2023)
\emph{On Isometries of Spectral Triples Associated to {$AF$}-algebras and Crossed Products}
\textit{J Noncommut Geom} 18(2):547-566 \href{https://doi.org/10.4171/JNCG/535}{https://doi.org/10.4171/JNCG/535}
\ 
\href{https://arxiv.org/abs/2301.07644}{arXiv:2301.07644}

\bibitem
{Contiquantum}\hfill

Bassi J, R. Conti R, Farsi C, Latr\'{e}moli\`{e}re F (2023)
\emph{Isometry Groups of Inductive Limits of Metric Spectral Triples and {Gromov-Hausdorff} Convergence} 
\textit{J London Math Soc} 108:1488 
\\ 
\href{https://doi.org/10.1112/jlms.12787}{https://doi.org/10.1112/jlms.12787}
\ 
\href{https://arxiv.org/abs/2302.09117}{arXiv:2302.09117}

\bibitem
{Halleffect}\hfill

Bellissard J, van Elst A, Schulz-Baldes H (1994)
\emph{The Noncommutative Geometry of the Quantum Hall Effect}
\textit{J Math Phys} 35(10):5373-5451 
\href{https://doi.org/10.1063/1.530758}{https://doi.org/10.1063/1.530758}
\ 
\href{https://arxiv.org/abs/cond-mat/9411052}{arXiv:cond-mat/9411052}

\bibitem
{Beggs-Majid-book}\hfill

Beggs E, Majid S (2010) 
\emph{Quantum Riemannian Geometry} 
\textit{Grundlehren der mathematischen Wissenschaften} 355
Springer 

\bibitem
{BGV}\hfill

Berline N, Getzler E, Vergne M (2004) 
\emph{Heat Kernels and Dirac Operators}
Springer 

\bibitem
{Paolo2005}\hfill

Bertozzini P, Conti R, Lewkeeratiyutkul W (2006)
\emph{A Category of Spectral Triples and Discrete Groups with Length Function}
\textit{Osaka J Math} 43(2):327-350 
\href{https://doi.org/10.18910/9728}{https://doi.org/10.18910/9728}
\ 
\href{https://arxiv.org/abs/math/0502583}{aXiv:math/0502583}

\bibitem
{BCL11}\hfill

Bertozzini P, Conti R, Lewkeeratiyutkul W (2011)
\emph{A Remark on {G}el'fand Duality for Spectral Triples}
\textit{Bull Korean Math Soc} 48(3):505-521
\href{https://doi.org/10.4134/bkms.2011.48.3.505}{https://doi.org/10.4134/bkms.2011.48.3.505}
\ 
\href{https://arxiv.org/abs/0812.3584}{arXiv:0812.3584}

\bibitem
{Bertozzini_2012}\hfill

Bertozzini P, Conti R, Lewkeeratiyutkul W (2012)
\emph{Categorical Non-commutative Geometry}
\textit{J Phys Conf Ser} 346:012003 
\href{https://doi.org/10.1088/1742-6596/346/1/012003}{https://doi.org/10.1088/1742-6596/346/1/012003}
\ 
\href{https://arxiv.org/abs/1409.1337}{arXiv:1409.1337}

\bibitem
{F}\hfill

Bertozzini P, Jaffrennou F (2013)
\emph{Remarks on Morphisms of Spectral Geometries},
\textit{East-West J of Math} 15(1):15-24 
\href{https://doi.org/10.48550/arXiv.1409.1342}{https://doi.org/10.48550/arXiv.1409.1342}
\ 
\href{https://arxiv.org/abs/1409.1342}{arXiv:1409.1342}

\bibitem
{Black}\hfill

Blackadar B (2006)
\emph{Operator Algebras} 
Springer

\bibitem
{borowiec}\hfill

Borowiec A (2024)
\emph{Quantum Calculi: Differential Forms and Vector Fields in Noncommutative Geometry} 
\textit{Scientific Legacy of Professor Zbigniew Oziewicz} 12:247-271
World Scientific 
\\ 
\href{https://doi.org/10.1142/9789811271151_0012}{https://doi.org/10.1142/9789811271151\_0012}
\ 
\href{https://arxiv.org/abs/2202.12102}{arXiv:2202.12102}

\bibitem
{CS}\hfill

Cannas Silva A (2008) 
\emph{Lectures on Symplectic Geometry}
Springer
   
\bibitem
{AF}\hfill

Christensen E, Ivan C (2006)
\emph{Spectral triples for {AF} C*-algebras and Metrics on the {Cantor} Set} 
\textit{J Oper Theory} 56:17-46 
\href{https://doi.org/10.48550/arXiv.math/0309044}{https://doi.org/10.48550/arXiv.math/0309044}
\ 
\href{https://arxiv.org/abs/math/0309044}{arXiv:math/0309044} 

\bibitem
{ConnesBook}\hfill 

Connes A (1994)
\emph{Noncommutative geometry}
Academic Press

\bibitem
{Connes}\hfill

Connes A (2013)
\emph{On the Spectral Characterization of Manifolds}
\textit{J Noncommut Geom} 7:1-82 
\href{https://doi.org/10.4171/jncg/108}{https://doi.org/10.4171/jncg/108}
\ 
\href{https://arxiv.org/abs/0810.2088}{arXiv:0810.2088}

\bibitem
{Twisted}\hfill

Dabrowski L, D’Andrea F, Magee A M (2021)
\emph{Twisted Reality and the Second-Order Condition}
\textit{Math Phys Anal Geom} 24(2):13 
\href{https://doi.org/10.1007/s11040-021-09384-4}{https://doi.org/10.1007/s11040-021-09384-4}
\ 
\href{https://arxiv.org/abs/1912.13364}{arXiv:1912.13364}

\bibitem
{Subman-intro}\hfill

Dajczer M, Tojeiro R (2019)
\emph{Submanifold Theory: Beyond an Introduction} 
Springer 
\\ 
\href{https://doi.org/10.1007/978-1-4939-9644-5}{https://doi.org/10.1007/978-1-4939-9644-5} 

\bibitem
{SIGMA16}\hfill

D'Andrea F (2020)
\emph{On the Notion of Noncommutative Submanifold}
\textit{SIGMA} 16:050 
\\ 
\href{https://doi.org/10.3842/sigma.2020.050}{https://doi.org/10.3842/sigma.2020.050}
\ 
\href{https://arxiv.org/abs/1912.01225}{arXiv:1912.01225}

\bibitem
{DV-I}\hfill

Dubois-Violette M (1988)
\emph{D\'erivations et Calcul Diff\'erentiel Non Commutatif}
\textit{C R Acad Sci Paris S\'er I Math} 307(8):403-408

\bibitem
{DV-M-II}\hfill

Dubois-Violette M, Michor P W (1994)
\emph{D\'erivations et Calcul Diff\'erentiel Non Commutatif. II.}
\textit{C R Acad Sci Paris S\'er I Math} 319(9):927-931
\href{https://doi.org/10.48550/arXiv.hep-th/9406166}{https://doi.org/10.48550/arXiv.hep-th/9406166}
\ 
\href{https://arxiv.org/abs/hep-th/9406166}{arXiv:hep-th/9406166}

\bibitem
{DV-M}\hfill

Dubois-Violette M, Michor P W (1996)
\emph{Connections on Central Bimodules in Noncommutative Differential Geometry} 
\textit{J Geom Phys} 20:218-232 
\href{https://doi.org/10.1016/0393-0440(95)00057-7}{https://doi.org/10.1016/0393-0440(95)00057-7}
\
\href{https://arxiv.org/abs/q-alg/9503020}{arXiv:q-alg/9503020}

\bibitem
{Emch}\hfill

Emch G G (1972)
\emph{Algebraic Methods in Statistical Mechanics and Quantum Field Theory}
{Wiley-Interscience}

\bibitem
{Flo-Gho}\hfill

Floricel R, Ghorbanpour A (2019)
\emph{On Inductive Limit Spectral Triples}
\textit{Proc Amer Math Soc} 147:3611-3619 
\href{https://doi.org/10.1090/proc/14583}{https://doi.org/10.1090/proc/14583}
\ 
\href{https://arxiv.org/abs/1712.09621}{arXiv:1712.09621}

\bibitem
{FGR}\hfill

Fr\"ohlich J, Grandjean O, Recknagel A (1998)
\emph{Supersymmetric Quantum Theory and Differential Geometry}
\textit{Commun Math Phys} 193(3):527-594 
\href{https://doi.org/10.1007/s002200050339}{https://doi.org/10.1007/s002200050339}
\ 
\href{https://arxiv.org/abs/hep-th/9612205}{arXiv:hep-th/9612205}

\bibitem
{GBVF}\hfill

Gracia-Bond{\'\i}a J M, V\'arilly J C, Figueroa H (2001)
\emph{Elements of Noncommutative Geometry}  
Birkh\"auser 

\bibitem
{Kaad-Suij}\hfill

Kaad J, van Suijlekom W D (2018) 
\emph{Riemannian Submersions and Factorization of {D}irac Operators}
\textit{J Noncommut Geom} 12(3):1133-1159 
\href{https://doi.org/10.4171/jncg/299}{https://doi.org/10.4171/jncg/299}
\ 
\href{https://arxiv.org/abs/1610.02873}{arXiv:1610.02873}

\bibitem
{Lee}\hfill

Lee J M (2013)
\emph{Introduction to Smooth Manifolds}
Springer 

\bibitem
{LORD2012}\hfill

Lord S, Rennie A, V\'{a}rilly J C (2012)
\emph{Riemannian Manifolds in Noncommutative Geometry} 
\textit{J Geom Phys} 62(7):1611-1638 
\href{https://doi.org/10.1016/j.geomphys.2012.03.004}{https://doi.org/10.1016/j.geomphys.2012.03.004}
\ 
\href{https://arxiv.org/abs/1109.2196}{arXiv:1109.2196}

\bibitem
{JMadore_1992}\hfill

Madore J (1992)
\emph{The Fuzzy Sphere}
\textit{Class Quant Grav} 9(1):69-88 
\\ 
\href{https://doi.org/10.1088/0264-9381/9/1/008}{https://doi.org/10.1088/0264-9381/9/1/008}

\bibitem
{Madore-book}\hfill

Madore J (1999) 
\emph{An Introduction to Noncommutative Differential Geometry and its Physical Applications} 
\textit{London Mathematical Society Lecture Note Series} 257  
{Cambridge University Press}

\bibitem
{Majid2000}\hfill

Majid S (2000)
\emph{Quantum Groups and Noncommutative Geometry} 
\textit{J Math Phys} 41(6):3892-3942 
\href{https://doi.org/10.1063/1.533331}{https://doi.org/10.1063/1.533331}
\ 
\href{https://arxiv.org/abs/hep-th/0006167}{arXiv:hep-th/0006167}

\bibitem
{Martinetti:2001fq}\hfill

Martinetti P, Wulkenhaar R (2002)
\emph{Discrete {K}aluza-{K}lein from Scalar Fluctuations in Noncommutative Geometry} 
\textit{J Math Phys} 43:182-204
\href{https://doi.org/10.1063/1.1418012}{https://doi.org/10.1063/1.1418012}
\ 
\href{https://arxiv.org/abs/hep-th/0104108}{arXiv:hep-th/0104108}

\bibitem
{doi:10.1063/1.531522}\hfill

Masson T (1996)
\emph{Submanifolds and Quotient Manifolds in Noncommutative Geometry},
\textit{J Math Phys} 37(5):2484-2497 
\href{https://doi.org/10.1063/1.531522}{https://doi.org/10.1063/1.531522}
\ 
\href{https://arxiv.org/abs/q-alg/9507030}{arXiv:q-alg/9507030}

\bibitem
{Mesland2012}\hfill

Mesland B (2012)
\emph{Spectral Triples and {$KK$}-theory: a Survey}
\textit{Clay Math Proc} 16:197-212 
\\ 
\href{https://doi.org/10.48550/arXiv.1304.3802}{https://doi.org/10.48550/arXiv.1304.3802}
\ 
\href{https://arxiv.org/abs/1304.3802}{arXiv:1304.3802}

\bibitem
{Mesland2014}\hfill

Mesland B (2012)
\emph{Unbounded Bivariant {$K$}-theory and Correspondences in Noncommutative Geometry} 
\textit{J f\"{u}r die Reine und Angew Math} 691:101-172

\href{https://doi.org/10.1515/crelle-2012-0076}{https://doi.org/10.1515/crelle-2012-0076}
\ 
\href{https://arxiv.org/abs/0904.4383}{arXiv:0904.4383}

\bibitem
{Nicolaescu-notes}\hfill

Nicolaescu L (2021)
\emph{The Co-area Formula} 
\textit{{N}otes for the ``{B}lue {C}ollar {S}eminar in {G}eometric {I}ntegration {T}heory''} 
\href{https://www3.nd.edu/~lnicolae/Coarea.pdf}{https://www3.nd.edu/~lnicolae/Coarea.pdf}

\bibitem
{MarioPaschke_2004}\hfill

Paschke M, Verch R (2004)
\emph{Local Covariant Quantum Field Theory over Spectral Geometries}
\textit{Class Quant Grav} 21(23):5299 
\href{https://doi.org/10.1088/0264-9381/21/23/001}{https://doi.org/10.1088/0264-9381/21/23/001}
\ 
\href{https://arxiv.org/abs/gr-qc/0405057}{arXiv:gr-qc/0405057}

\bibitem
{Peter}\hfill
 
Pedersen P (2016)
\emph{Riemannian Geometry}
Springer 

\bibitem
{chatchai}\hfill

Puttirungroj C (2021)
\emph{Horizontal Categorification of Spectral Triples} 
PhD Thesis Thammasat University 
\href{https://doi.org/10.14457/TU.the.2021.472}{https://doi.org/10.14457/TU.the.2021.472} 

\bibitem
{AlmostNCG}\hfill

van Suijlekom W D (2015)
\emph{Noncommutative geometry and particle physics}
Springer 

\bibitem
{Suij-Verh}\hfill

van Suijlekom W D, Verhoeven L S (2022)
\emph{Immersions and the Unbounded {K}asparov Product: Embedding Spheres into {Euclidean} Space}
\textit{J Noncommut Geom} 16(2):489–511 
\\ 
\href{https://doi.org/10.4171/jncg/451}{https://doi.org/10.4171/jncg/451}
\ 
\href{https://arxiv.org/abs/1911.06044}{arXiv:1911.06044}

\bibitem
{Suij-Verh2}\hfill

van Suijlekom W D, Verhoeven L S (2023)
\emph{Riemannian Embeddings in Codimension One as Unbounded {$KK$}-cycles}
\textit{Ann {$K$}-Theory} 8(4):645-668 
\href{https://doi.org/10.2140/akt.2023.8.645}{https://doi.org/10.2140/akt.2023.8.645}
\ 
\href{https://arxiv.org/abs/2212.08053}{arXiv:2212.08053}

\end{thebibliography}
}
\end{document}